\documentclass[11pt,english]{article}
\usepackage[latin9]{inputenc}
\usepackage{geometry}
\usepackage{verbatim}
\usepackage{parskip}
\usepackage{float}
\usepackage{amsthm}
\usepackage{amsmath}
\usepackage{amssymb}
\usepackage{mathtools}
\usepackage{enumitem}
\usepackage[dvipsnames]{xcolor}
\usepackage[backref=page,linktocpage=true,breaklinks,colorlinks,citecolor=PineGreen,linkcolor=Violet]{hyperref}
\usepackage{cleveref}
\usepackage{fullpage}
\usepackage{tikz}
\usepackage{forest}
 \usetikzlibrary{positioning,decorations.pathreplacing}
 \usetikzlibrary{calc,arrows.meta,shapes.geometric,trees}
 \usepackage{subcaption}

\makeatletter

\floatstyle{ruled}
\newfloat{algorithm}{tbp}{loa}
\providecommand{\algorithmname}{Algorithm}
\floatname{algorithm}{\protect\algorithmname}

  \newtheorem{theorem}{Theorem}
  \newtheorem{lemma}[theorem]{Lemma}
\newtheorem{claim}[theorem]{Claim}
\numberwithin{theorem}{section}
  \numberwithin{equation}{section}
  \theoremstyle{definition}
  \newtheorem{defn}[theorem]{\protect\definitionname}
  \newtheorem{example}[theorem]{Example}

  \newtheorem{observation}[theorem]{Observation}

  \theoremstyle{remark}

\theoremstyle{plain}
\newtheorem{thm}[theorem]{\protect\theoremname}
  \theoremstyle{plain}
  
  \theoremstyle{plain}

\usepackage{complexity,color,bm}
\usepackage{colortbl}

\newcommand{\Alg}{{\mathsf{Alg}}}
\providecommand{\Opt}{\ensuremath{\mathsf{Opt}}}
\providecommand{\poly}{\ensuremath{\mathsf{poly}}}

\newcommand{\calV}{{\cal V}}
\providecommand{\R}{{\mathbb{R}}}
\providecommand{\E}{{\mathbb{E}}}

\newcommand{\mF}{\mathcal{F}}
\newcommand{\ignore}[1]{}

\newcommand{\ot}{\overline{t}}

\definecolor{gray-comment}{gray}{0.5}

\theoremstyle{plain}

\makeatletter
\newtheorem*{rep@theorem}{\rep@title}
\newcommand{\newreptheorem}[2]{%
\newenvironment{rep#1}[1]{%
 \def\rep@title{#2 \ref{##1}}%
 \begin{rep@theorem}}%
 {\end{rep@theorem}}}
\makeatother

\newreptheorem{rcor}{Corollary}
\newreptheorem{rthm}{Theorem}

\usepackage{thmtools}
\usepackage{thm-restate}
\makeatother

\usepackage{babel}
  \providecommand{\definitionname}{Definition}
  \providecommand{\lemmaname}{Lemma}
  \providecommand{\propositionname}{Proposition}
\providecommand{\theoremname}{Theorem}

\newcounter{note}[section]

\title{Secretary, Prophet, and Stochastic Probing via Big-Decisions-First}

\author{ Aviad Rubinstein\thanks{(aviad@cs.stanford.edu) Department of Computer Science, Stanford University. Supported in part by David and Lucile Packard Fellowship.} \and Sahil Singla\thanks{
        (ssingla@gatech.edu)
        School of Computer Science,
        Georgia Tech.
        Supported in part by NSF awards CCF-2327010 and CCF-2440113.
        }}
\date{\today}

\begin{document}
\maketitle

%\snote{Old title: Big Decisions First: Near-Optimal Bounds for Secretary, Prophet, and Stochastic Probing  Beyond Binary Values}

\begin{abstract}
  We revisit three fundamental problems in algorithms under uncertainty: the Secretary Problem, Prophet Inequality, and Stochastic Probing, each subject to general downward-closed constraints. When elements have binary values, all three problems admit a tight $\tilde{\Theta}(\log n)$-factor approximation guarantee. 
    For general (non-binary) values, however, the best known algorithms lose an additional $\log n$ factor when discretizing to binary values, leaving a quadratic gap of $\tilde{\Theta}(\log n)$ vs.\ $\tilde{\Theta}(\log^2 n)$.
    \medskip

    We resolve this quadratic gap for all three problems, showing $\tilde{\Omega}(\log^2 n)$-hardness for two of them and an $O(\log n)$-approximation algorithm for the third. While the technical details differ across settings, and between algorithmic and hardness proofs, all our results stem from a single core observation, which we call the \emph{Big-Decisions-First Principle:}
     Under uncertainty, it is better to resolve high-stakes (large-value) decisions early.
%In the face of uncertainty, it is better to make the big, high-stake decisions first, and defer smaller decisions for last.  
    %\end{quote}
\end{abstract}

%\bigskip

% \setcounter{tocdepth}{1}
%  {\small
%  %\begin{spacing}{0}
%     \tableofcontents
%     \thispagestyle{empty}
%  %\end{spacing}
%  }

%TO DO LIST
% \begin{enumerate}
    %\item Abstract (1st draft done - Sahil to make next pass)
 %   \item Sec 2 Techn Overview pass (Sahil): Made a pass, might revist again
%    \item Secretary Pb section pass (Sahil): Working on now
 %   \item Formal problem definitions (Sahil): Added. Aviad take a pass
 %   \item Related works + discussion of XOS vs norms vs downward closed constraints (Sahil)
  %  \item Resolving outstanding comments (both)
  %  \item Error correcting code discussion (Aviad): Added - Sahil to take a pass! 
   %  \item Intro (Sahil made another pass, next Aviad's pass)
 %    \item Combined figure without ECC for PI (Sahil): I made some changes. I am not sure what we want to change/add. 
    %\item  \Aviad{should we remove the newpages between sections?}
 %\end{enumerate}

% \setcounter{page}{0}
 \clearpage

 \section{Introduction}

% We study three fundamental problems in algorithm design under uncertainty. Each has a rich history in theoretical computer science, and all share a common core: we must select a feasible subset of $n$ elements with uncertain values to maximize total value. We briefly recall these classical problems before describing our results \Aviad{Formal definitions are deferred to the respective technical sections}.

We study three fundamental problems in algorithm design under uncertainty, where the goal is to select a feasible subset of $n$ elements with uncertain values to maximize total value. %, each has a rich history in theoretical computer science.
We briefly recall these classical problems before describing our results; formal definitions are deferred to the respective technical sections. %The objective is common to all three problems: 

 \textbf{Secretary Problem.} Elements (``secretaries'') have adversarially chosen values but arrive in a uniformly {\em random order}. Upon each arrival, the algorithm observes each element and must decide immediately and irrevocably whether to accept it. 

 \textbf{Prophet Inequality.} Element values are drawn from {\em known independent distributions}, but the elements arrive in an {\em adversarial order}. 
 As in the secretary problem, the algorithm observes each element and must decide immediately and irrevocably whether to accept it. 
 %\Aviad{why would we use different language here instead of highlighting the parallel to secretary?} The algorithm again must decide online, with no recall, whether to keep each revealed value.

 \textbf{Stochastic Probing.} We face two layers of feasibility constraints: the {\em outer constraint} specifies which elements can be ``probed'' to reveal their random values (drawn independently from known distributions); the {\em inner constraint} governs which of the probed elements can finally be selected. The goal is to compare the best adaptive policy to the best non-adaptive one (the \emph{adaptivity gap}).

All three problems are considered under arbitrary \emph{downward-closed feasibility constraints} (i.e., every subset of a feasible set is feasible).
For binary-valued
(unweighted)  instances, i.e.~where each element's value is either $0$ or $1$, these problems are essentially settled: closely related algorithms guarantee $O(\log n)$-approximations~\cite{Rubinstein16,LiLZ25}, and essentially the same counterexample~\cite{BabaioffIK07,Rubinstein16,LiLZ25} shows that this is almost tight.
%For binary-valued (unweighted)  instances, where each element's value is either $0$ or $1$, these problems are essentially settled: algorithms and lower bounds meet at $\tilde{\Theta}(\log n)$-approximations ($\tilde{\Theta}(\cdot)$ hides $\poly(\log\!\log n)$ factors) \cite{Rubinstein16, LiLZ25, BabaioffIK07}.
However, for general (weighted) values, the previously best known approach used a simple reduction to unweighted case by rounding to binary values at a random threshold from $T := \{1,1/2,1/2^2,\dots,1/\poly(n)\}$, % $\poly(n)$ factor range.
which loses an extra $|T| = \Theta(\log n)$  factor. %the instance to the binary case by thresholding values using a random threshold%. % discretizing values into $O(\log n)$ %``power-of-two'' \Aviad{``power of two'' is something different in online algorithms ;) I guess you meant ``powers of two''.} \Aviad{I think it's misleading to say that the reduction is by discretizing and that's why they lose a log factor, because we also discretize} levels
%, incurring an extra $O(\log n)$ loss. \snote{I think it's useful to mention here that the $\log n$ factor is not magical, but a simple trick of reducing to log levels after discretizing and choosing a random level; my previous phrasing was maybe confusing, it was not about discretizing but about the simple trick (often called "powers of 2" trick).} 
This leaves the same quadratic gap open across all three problems:

\begin{quote}
    {\bf Main Question:} Can we close the $\tilde{\Theta}(\log n)$ vs.\,$\tilde{\Theta}(\log^2 n)$ gap in approximation factors for the Secretary Problem, Prophet Inequality, and Stochastic Probing under arbitrary downward-closed constraints?
\end{quote}
This question has been posed explicitly in multiple venues, e.g.~by Bobby Kleinberg for the Secretary and Prophet settings (Open problem session at Simons Institute program on Algorithms and Uncertainty, Fall 2016) and recently by Li, Lu, and Zhang~\cite{LiLZ25} for Stochastic Probing, who note that ``bridging this gap likely requires deeper understanding of non-binary weight vectors''.

%Resolving the quadratic gap was mentioned as an explicit open problem, e.g.~by Bobby Kleinberg for Secretary Problem and Prophet Inequality\footnote{Open problem session at Simons Institute Fall 2016 program on Algorithms and Uncertainty.}, and more recently by~\cite{LiLZ25} for  Stochastic Probing% \footnote{``bridging this gap likely requires deeper understanding of the non-binary weight vectors in the XOS representation of the objective function $f$, representing an important direction for future research.'' Here, XOS representation means the same as general downward-closed inner constraints.}.

\subsection{Our Contributions} 
We resolve this quadratic gap, giving nearly tight (up to $\poly(\log\!\log n)$) upper and lower bounds for all three settings. What makes our findings particularly striking is that the Secretary and Prophet settings---which often behave similarly up to $O(1)$ factors---turn out to diverge sharply once general (non-binary) values are considered. 
Even more unexpectedly, both algorithmic and hardness results for all three problems stem from the same underlying principle:

%We are able to resolve this quadratic gap and obtain near-optimal (up to $\poly(\log\!\log n)$ factors) answers in all the three settings by establishing new upper and lower bounds. We find these answers, and in particular the contrast between the Secretary Problem and Prophet Inequality settings quite surprising. But perhaps more surprising is that all our results, including both upper and lower bounds, were derived   from the same core fundamental observation:
%, which we call {\em Big Decisions First Principle}:

\begin{quote}
    {\bf Big-Decisions-First Principle.} Under uncertainty, it is better to resolve high-stakes (large-value) decisions early, and defer smaller decisions until later.
%In the face of uncertainty, it is better to make the big, high-stake decisions first, and defer smaller decisions for last.  
\end{quote}

In our setting, this principle says that one should ideally handle large-valued elements first, selecting or probing them before smaller ones, so that uncertainty about ``big'' decisions is resolved early.
Surprisingly, this single idea explains the full range of behavior across the three classical models. The key to applying this principle to each problem is to understand to what extent each setting allows the algorithm to approximate the ideal big-values-then-low-values order.

%\vspace{-0.1cm}
%\begin{enumerate}[label=(\roman*),topsep=0pt,itemsep=-0.5ex,partopsep=1ex,parsep=1ex,leftmargin=0.75cm] 

%\begin{description}
   \paragraph{Worst-case order --- Prophet Inequality.} In the Prophet Inequality setting, the elements arrive in worst-case order. We prove a nearly-tight $\tilde{\Omega}(\log^2 n)$ lower bound by constructing an instance that perfectly reverses the ideal algorithm that starts  from high value elements.

\begin{theorem}[Informal \Cref{thm:PIHardness}]\label{thm:prophet_main}
For Prophet Inequality with arbitrary downward-closed constraints and an additive objective, there exist instances where no online algorithm can achieve in expectation better than an $\tilde{\Omega}(\log^2 n)$ approximation.
\end{theorem}

Of course, even with the worst-case order of arrival of elements, it's not obvious how to construct the combinatorial constraints that connect low and high value elements.
Note that we cannot simply concatenate $\log n$ unrelated binary hardness instances for each of the different values because the algorithm could solve each of them separately. To get around this, we present a ``double error correcting code'' construction that ties  the different components together, ensuring that selecting even a small fraction of the low value elements heavily constrains the algorithm's choices among high value elements. (See Section~\ref{sec:overview} for an overview and Section~\ref{sec:prophet} for details.)

%Since $\tilde{O}(\log^2 n)$ upper bounds are already known for both additive \cite{Rubinstein16} and XOS \cite{RubinsteinS17} objectives, this establishes that the correct bound for Prophet Inequality is $\tilde{\Theta}(\log^2 n)$.
%\Aviad{we haven't introduced XOS yet}

\paragraph{Best-case order --- Stochastic Probing.}
In the Stochastic Probing problem, we establish an $\tilde{\Omega}(\log^2 n)$ adaptivity gap by constructing an instance where the adaptive policy can visit the elements in the ideal order of high-value first.

% The Stochastic Probing setting provides the ``mirror image'' of Prophet Inequality.
% An adaptive policy can explicitly implement the {Big Decisions First} strategy, probing high-value elements first and deciding later about smaller ones.
% We prove this adaptivity advantage is again quadratic, matching the Prophet separation in magnitude but in the opposite direction.

\begin{theorem}[Informal \Cref{thm:adapGapLB}]\label{thm:probing_main}
For Stochastic Probing with arbitrary downward-closed outer and inner constraints, there exist instances where every non-adaptive algorithm achieves at most an  $\tilde{\Omega}(\log^2 n)$ fraction of the optimal adaptive value.
\end{theorem}

%\snote{Since next Sec on techniques is long, do we want to write a couple of lines on our hardness techniques here,  highlighting double-error correcting codes.}

\paragraph{Secretary Problem --- Average-case order.}
% Finally, the Secretary Problem represents the ``average-case'' middle of this spectrum.
% The random arrival order allows us to approximately simulate the ideal ``decreasing-by-value'' sequence.
% Our algorithm partitions the stream into $\tilde{\Theta}(\log n)$ value blocks (powers of two), processes these blocks in decreasing order, and achieves an $O(\log n)$ approximation.
Should the uniformly random order of the secretary problem look more like the worst-case or best-case scenario? We design an algorithm that partitions the stream of elements into $\tilde{\Theta}(\log n)$ blocks, and each block only considers elements from a single value (rounded to nearest power of $2$). Across blocks the element-values are considered in decreasing order, mimicking the ideal scenario and guaranteeing an $O(\log n)$-factor approximation for the problem with general values. In fact, we're also able to extend this algorithm, with the same approximation guarantee, to XOS valuation functions which allows for interaction between values of selected elements~\cite{RubinsteinS17}.

\begin{theorem}[Informal \Cref{thm:mainSecretFormal}]\label{thm:sec_main}
For the Secretary Problem with arbitrary downward-closed constraints and an additive (or even submodular/XOS) objective, there exists an online algorithm achieving an $O(\log n)$-approximation to the offline optimum.
\end{theorem}

Together with the $\tilde{\Omega}(\log n)$ hardness of \cite{BabaioffIK07}, this settles the Secretary Problem at $\tilde{\Theta}(\log n)$.
It is interesting to note  that random order behaves as well as the ``best possible'' order up to $\poly(\log\!\log(n))$ factors.%, even for general (non-binary) values.

%This separation between the three models is conceptually striking.
% For all natural classes of downward-closed constraints studied so far (like partition, graphic, laminar, and even matroid), the Secretary, Prophet, and Probing problems are known or conjectured to match within constant factors.
% Our results show that under fully general downward-closed constraints, this equivalence breaks sharply: $\tilde{\Theta}(\log n)$ for Secretary versus $\tilde{\Theta}(\log^2 n)$ for Prophet and Stochastic Probing.
% Together, these form a clean spectrum from average-case (Secretary) to worst-case (Prophet) to best-case (Probing) order, unified by one guiding principle: \emph{who gets to make the big decisions first.}

%\end{description}
%\end{enumerate}   

 \subsection{Related Work}

\textbf{Secretary and Prophet Inequality.}
These problems originated in optimal stopping theory \cite{Dynkin, KrengelS77, KrengelS78}. They have been deeply explored in TCS in the last two decades, due to their applications in algorithmic mechanism design and beyond worst-case analysis of online algorithms.  We refer the interested readers to survey \cite{Lucier17} and book chapters \cite{GS-Book20} and \cite[Chapter 30]{EIV-Book23}, and below we only discuss the results on general downward-closed constraints.

Secretary and Prophet inequality problems for general downward-closed constraints were first proposed in \cite{BabaioffIK07}. This paper, which also proposed the famous matroid secretary conjecture,  showed an $\tilde{\Omega}(\log n)$-hardness for general downward-closed constraints  (see journal version \cite{babaioff2018matroid}). Rubinstein \cite{Rubinstein16} designed matching $\tilde{O}(\log n)$-competitive algorithm for both secretary and prophet problems with binary values. Subsequently, \cite{RubinsteinS17} generalized this problem to XOS/subadditive combinatorial functions, and obtained $\poly(\log n)$-competitive algorithms. In particular, they observed that  downward-closed constraints can capture binary XOS functions (i.e., where clauses have $\{0,1\}$ value coefficients) by creating a feasible subset $A$ whenever there is a clause with value $1$ for all elements in $A$. However, both the approaches in \cite{Rubinstein16} and \cite{RubinsteinS17} only work for binary values, leaving the $\tilde{\Theta}(\log n)$ vs.\,$\tilde{\Theta}(\log^2 n)$ gap open for additive and XOS objectives.

 %\snote{talk about R. limited to bernoulli.  talk about RS idea to interpret XOS as another down-closed constraint. } 
 
\medskip
\textbf{Stochastic Probing.} 
This problem with outer and inner constraints was first formulated by Gupta and Nagarajan \cite{GN-ICALP13}, as a generalization of  stochastic   matching  \cite{CIKMR-ICALP09} and Bayesian mechanism design \cite{CHMS-STOC10}. In particular, they showed that the adaptivity gap is $O(k)$ when both the outer and inner constraints are intersections of $k$ matroids\footnote{In the original formulation \cite{GN-ICALP13}, the algorithm is always ``committed'' to take any probed active element, but subsequent work used ``stochastic probing'' to mean without commitment unless explicitly stated. We remark that our lower bound instance on adaptivity gap achieves ``best of both worlds'' with respect to commitment: it shows the separation even if the adaptive policy is constrained to be committed and the non-adaptive policy is uncommitted.}. Separately, \cite{ANS-WINE08} showed $O(1)$ adaptivity gaps when the inner constraint is a polymatroid (i.e., the objective function is submodular)  and the outer constraint is a matroid.
In a sequence of works \cite{BSZ-RANDOM19,GNS-SODA17,GNS-SODA16}, both these results were greatly generalized to show $O(1)$-adaptivity gaps for general downward-closed outer constraints and polymatroid inner constraints. 

The main open question left in the above line of work, originally due to \cite{GNS-SODA17}, was whether the adaptivity gap is at most $\poly(\log n)$  when both the outer and inner constraints are general downward-closed (i.e., the objective function is XOS). 
This conjecture was reformulated in terms of general norms (which are continuous analog of XOS functions as both are max-of-linear functions) in \cite{PRS-APPROX23,KMS-STOC24}, who made progress for special classes of norms such as symmetric and $p$-supermodular norms. Finally, the recent paper \cite{LiLZ25} achieved a breakthrough by showing $\poly(\log n)$ adaptivity gap when both the outer and inner constraints are general downward-closed. Although \cite{LiLZ25} obtains the tight $\tilde{\Theta}(\log n)$ adaptivity gap for binary value distributions, they left open the $\tilde{\Theta}(\log n)$ vs.\,$\tilde{\Theta}(\log^2 n)$ gap for general value distributions.

 \section{Technical Overview}\label{sec:overview}

 We begin by discussing prior approaches and then describe how to use our Big-Decisions-First Principle to obtain improved algorithmic and hardness bounds.

\medskip
\textbf{Prior Approaches for Binary Values.} At a high-level, the core idea behind previous algorithms \cite{Rubinstein16,RubinsteinS17,LiLZ25} for binary values  is the same\footnote{One exception is that \cite{Rubinstein16} gives a completely different algorithm for prophet inequality; but the same guarantee can also be obtained from \cite{Rubinstein16}'s secretary algorithm using \cite{AzarKW19}'s reduction from order-oblivious algorithms for secretary to (sample-based) prophet inequality.}:  we obtain (or hallucinate) a sample from the underlying distributions and make ``greedy''   decisions assuming that the sample represents the entire distribution. 
For concreteness, think of  secretary problem with an  additive binary objective. Here,  we view the first half of the stream  as ``sample'' elements to learn the distribution (no sample elements are selected), and then each arriving element is   selected  if it is feasible and simultaneously combines well with (a)~the set of elements selected so far and (b)~an optimal subset of the sample.  

 To analyze this algorithm, observe that at first the sample is representative of the entire distribution; but as the algorithm selects more elements, it is over-fitting to the sample. To quantify the level of over-fitting,   consider any possible \emph{fixed set} of elements  $A$ that could be selected  by the algorithm. It is possible to argue that the sample gives us an estimate on the compatibility of other elements with set $A$ with tail probability as low as $\approx 2^{-\Omega(|\Opt|)}$. 
 As long as our algorithm never attempts to select more than  $|A| \le |\Opt|/\log n$ elements, since there are only $n^{O(|\Opt|/\log n)} = 2^{O(|\Opt|)}$ such subsets, and therefore we can take a union bound over all of them to account for over-fitting. %Hence, our algorithm can only output a subset of size $\Omega(|\Opt|/\log n)$.

%  To analyze this algorithm, observe that at first the sample is representative of the entire distribution; but as the algorithm selects more elements, it is over-fitting to the sample. More formally,   consider any possible \emph{fixed set} of elements  $A$ that could be selected  by the algorithm. It is possible to argue that the sample gives us an estimate on the compatibility of \Aviad{other elements with set $A$?} set $A$ with other elements  with tail probability as low as $\approx 2^{-\Omega(|\Opt|)}$. Thus, \Aviad{``our algorithm will never output any small subset  $A$ of size $|A| \ll |\Opt|/\log n$'' doesn't follow from the previous sentence though!} our algorithm will never output any small subset  $A$ of size $|A| \ll |\Opt|/\log n$ since there are only $n^{O(|\Opt|/\log n)} = 2^{O(|\Opt|)}$ such subsets, and we can take a union bound over all of them to account for over-fitting. Hence, our algorithm can only output a subset of size $\Omega(|\Opt|/\log n)$.

\medskip  \textbf{Challenge in Going Beyond Binary Values.} Suppose we partition the elements   into $O(\log n )$ buckets based on their approximate value, where $W_C$ denotes the elements of value $C$.
In a typical worst case instance, each bucket has roughly the same total value, so buckets with lower values (per element) have more elements.
Suppose for simplicity that we only have two buckets: $W_H, W_L$, with high- and low-value elements, respectively, and let $\Alg_H,\Alg_L$ denote the respective subsets of elements selected by our algorithm. If we try to directly apply the greedy approach of the unweighted algorithm to weighted instances, the number of $W_L$ elements selected by the algorithm can quickly dwarf the total number of $W_H$ elements in the entire instance: $|\Alg_L| \gg |W_H|$. 
Now the concentration for $W_H$ elements has tail probability at least $2^{-|W_H|}$, so we can no longer take a union bound over all $n\choose{|\Alg_L|}$ choices of low elements.
%
%The challenge in directly applying the  algorithm for binary values to general values is that  the tail probability for $W_C$ elements  only holds like  $2^{-\Omega(|\Opt \cap W_C|)}$, i.e., exponentially in the number of elements of value $C$  contributing to the optimum. However, since the  algorithm might now be selecting a subset of many different values, it is no longer sufficient to only rule out small subsets $A$ of size $\ll |\Opt \cap W_C|/\log n$. Hence,  

Instead of the greedy approach, prior works simply  
reduces to the unweighted case by optimizing only for a single (random or optimal) value $C$.  In the worst case when there are $\log n$ buckets with roughly equal total contribution to the optimum, this naive approach loses another $\log n$ factor.

% \medskip  \textbf{Challenge in Going Beyond Binary Values.} Suppose we partition the elements   into $O(\log n )$ buckets based on their approximate value, where $W_C$ denotes the elements of value $C$. 
% %
% The challenge in directly applying the  algorithm for binary values to general values is that  the tail probability for $W_C$ elements  only holds like  $2^{-\Omega(|\Opt \cap W_C|)}$, i.e., exponentially in the number of elements of value $C$  contributing to the optimum. However, since the  algorithm might now be selecting a subset of many different values, it is no longer sufficient to only rule out small subsets $A$ of size $\ll |\Opt \cap W_C|/\log n$. Hence,  prior works would 
% optimize only for a single (random or optimal) value $C$.  In the worst case, each bucket has roughly equal total contribution to the optimum, so this approaches loses another $\log(n)$ factor. 

\medskip  
\textbf{Core Observation: Best-Case Order (``Big-Decisions-First'').}  
 Suppose that we could process the buckets $W_C$ in order of {\em increasing number of elements in $\Opt$}, i.e.~increasing $|\Opt \cap W_C|$. (In the aforementioned worst case where the total contribution per bucket is equal, this is equivalent to order of {\em decreasing element values}.) A key observation 
 to all our results is that in this wishful case, the cost of union bound over choice of smaller subsets doesn't affect the algorithm's performance on large subsets; so we could get an $\Omega(1/\log n)$-fraction of the optimal contribution from every bucket $W_C$ instead of optimizing for only one bucket. This technical observation forms the basis of our Big-Decisions-First Principle, and will guide both our secretary algorithm and hardness instances for prophet inequality and stochastic probing.

\subsection{Secretary Algorithm}

 %\medskip  \textbf{Algorithm for Secretary Problem.}
 Our $O(\log n)$-approximation algorithm for the secretary problem relies on one more key observation, besides our core observation  above. The  observation  is that in the binary case, the algorithm of~\cite{Rubinstein16} already obtains its entire $\Omega(1/\log n)$-approximation after processing only $(1/\log n)$-fraction of the ``real'' (non-sample/second-half) elements.
 %: As outlined above, until that point the algorithm is guaranteed to make progress at a near-optimal rate. 
Hence, at a high level,  our secretary algorithm  partitions the stream into $O(\log n)$ intervals, one for each bucket in decreasing order of value $C$. For each interval, the algorithm uses a (variant of) \cite{Rubinstein16}'s algorithm to collect a near-optimal subset of elements from the corresponding value bucket, while (i) respecting~the previous selections of higher-value elements, and (ii) optimizing for a matching subset of lower-value sample elements. As per our core observation, the sample of lower-value elements will tend to be a representative sample, even after a union bound over all selections of larger-value elements. 

 The final algorithm carefully balances a few more mundane but important details (see Section~\ref{sec:secretary}). For example, when extending this algorithm to XOS functions, we need to be careful when bucketing elements because their marginal values are different for different clauses. We also need to be careful to set up the analysis to take advantage of the strong concentration before partitioning into intervals --- even though we only need $O(1/\log n)$-fraction of the stream to actually collect values, we need the full stream (or at least $\Omega(1)$-fraction) to achieve optimal concentration.

\subsection{Hardness for Prophet Inequality}
In light of our core observation, the key to showing hardness for prophet inequality is to force the algorithm to process elements from the lowest value bucket to the highest, with (i)~roughly equal total contributions from each bucket, and (ii)~a (variant of) the tight $\tilde{\Omega}(\log(n))$ lower bound for binary values from~\cite{BabaioffIK07,Rubinstein16} inside each bucket. 

Let us quickly recall the basic hard instance for the binary case: Fix some $L = \tilde{\Theta}(\log(n))$ such that $L^L \ll n/L$. The hard instance of~\cite{BabaioffIK07,Rubinstein16} has $n/L$ disjoint feasible sets, each of size $L$; elements have nonzero value with probability $1/L$. In hindsight, the best of $n/L \gg L^L$ subsets will have value $\Omega(L)$, but as soon as the online algorithm selects a single element, it commits to one set --- and the rest of the set only contributes $O(1)$ in expectation. 

\medskip \textbf{Construction of the Hard Instance.}
To show $\Omega(L^2)$-hardness  beyond binary values, we  think of our hard instance distribution as an $L$-layered rooted tree,  where each vertex at distance $\ell>0$ from the root corresponds to a subset of elements of value $\in \{0,L^{-\ell}\}$. The feasible sets are the root-leaf paths (the union of subsets corresponding to all the vertices on such paths). Each element has nonzero value with probability $1/L^2$, and the elements in layers closest to the leaves arrive first,  forcing the online algorithm to make the small decisions first. See also Figure~\ref{fig:prophet-LB}.

 In the leaves layer, vertex-sets have $L^L$ elements, so to lower bound the prophet's offline $\Opt$ we would need to have $L^{\Omega(L^L)} \approx L^n $ vertex-sets to obtain a vertex-set with $\Omega(1)$-fraction of elements with  non-zero values. Since we only have a total of $n \ll L^n$ elements (across all layers), we can no longer use disjoint feasible subsets. Instead, we allow the subsets to intersect, and use an error-correcting code construction to guarantee that the intersection is small.

 The crux of our construction is in combining the different layers into one instance. To give the algorithm the disadvantage a-la Big-Decisions-First, we would like the property that selecting a large subset of small elements forces the algorithm to commit to a unique path to the root. And once the algorithm commits to a specific path to the root, each element in it is only nonzero with probability $1/L^2$ (so a total of $1/L^2$-fraction of $\Opt$ in expectation). In fact, we want the stronger property that even if the algorithm selects just a $1/L^2$-fraction of a given vertex-subset, it is already committed to the entire path to the root. We obtain this property using a ``double error-correcting code'' construction, which simultaneously guarantees: \begin{enumerate}
    \item high distance to all other subsets at the same layer; and
    \item even higher distance to subsets that don't share the same path from the root\footnote{This latter property is similar to ``list-decodable codes'' where up to the higher distance we can only guarantee to ``decode'' the subset to the ``list'' of subsets with the same root-parent path. (But note that for our purposes a simple construction based on random codes suffices.)}.
\end{enumerate} 

% Put this in your LaTeX preamble:
% \usepackage{tikz}

\begin{figure*}[t]
    \centering

    \scalebox{1.15}{
    %\begin{center}
%\centering
\begin{tikzpicture}[
  grow=down,
  level distance=12mm,
  level 1/.style={sibling distance=28mm},
  level 2/.style={sibling distance=18mm},
  blk/.style={draw, rounded corners, rectangle, minimum width=8mm, minimum height=5mm},
  edge from parent/.style={draw},
  edge from parent path={(\tikzparentnode.south)--(\tikzchildnode.north)}
]

% Depth 1 (root) → Depth 3 leaves
\node[blk] (root){\footnotesize Root (no subset)} % depth 1 (root)
  child { node[blk] {} % depth 2
    child { node[blk] (b1) {} } % depth 3
    child { node[blk] {} }
  }
  child { node[blk] {} % depth 2
    child { node[blk] {} } % depth 3
    child { node[blk] {} }
  }
  child { node[blk] (a2) {} % depth 2
    child { node[blk] {} } % depth 3
    child { node[blk] {} }
  };

% Optional: annotate branching factors (pure labels; no edges)

%\node[align=left, anchor=west] at ($(root.east)+(0.3,0.0)$) {\small root at $\ell=0$ has $(L^{2L})^{d_{1}}$ children};

%\node[align=left, anchor=west] at ($(a2.east)+(0.3,0.0)$) {\small each $\ell=1$ node has $(L^{2L})^{d_{2}}$ children};

\draw[-{Stealth[length=3.2mm,width=3.2mm]}, line width=1.2pt]
(-4.5,-2.4) -- (-4.5,0.5);

\node[align=left, anchor=west, rotate=90] at ($(b1.east)-(2.2,0)$) {order of arrival,\\ increasing values,\\ decreasing set sizes};

%\node[bottom, left=13pt, align=left, text width=3.2cm, rotate=90]{\tiny order of arrival}

\end{tikzpicture}
%\end{center}
}
    
     \raggedright
\small 
Feasible sets are unions of subsets along any root-leaf path. Subsets at different layers are disjoint, but same layer subsets may intersect---they share at most $1/L$ fraction of the elements if they have the same parent and at most $1/L^2$ fraction of the elements if their parents are different.

    \caption{Prophet inequality construction}
    \label{fig:prophet-LB}
\end{figure*}

\textbf{Overview of Analysis.}
To analyze this construction, we bound the contribution from three parts of any feasible set (path) chosen by the algorithm. Let $\ell^*$ be the depth (distance from root) of the last vertex for which  the algorithm obtains an $\omega(1/L^2)$-fraction of the corresponding subset.
\begin{description}
    \item[Layers $\ell' > \ell^*$.] By definition, the algorithm obtains at most $O(1/L^2)$-fraction of the feasible elements from these layers.
    
    \item[Layer $\ell^*$.] Similar to the analysis of the binary hard case, but also using the same layer intersection of error-correcting code construction, for this layer the algorithm obtains at best an $O(1/L)$-factor approximation, or $O(1/L^2)$-fraction of the total value of the entire instance.
    
    \item[Layers $\ell' < \ell^*$.] The elements in this layer appear after the algorithm selected $\omega(1/L^2)$-fraction of a layer-$\ell^*$ set. By the double error-correcting code construction, the algorithm has  committed to a unique path to the root, where each element is nonzero only with probability $1/L^2$.
\end{description}
In total, each of the three cases contributes at most $O(1/L^2)$-fraction of the total value.

\subsection{Large Adaptivity Gap for Stochastic Probing}
The intuition for stochastic probing is similar to prophet inequality, but here we use an optimal order instead of worst-case order: Recall that our goal is to lower bound the adaptivity gap, so we allow the adaptive benchmark to probe the high value elements first, then make decisions about lower value elements --- this gives it another advantage over the non-adaptive policy. The actual construction of the hard instance builds on and generalizes the construction for prophet inequality.

%The intuition for stochastic probing is at a high level the opposite of prophet inequality (best order vs worst case order). Specifically, we allow the adaptive benchmark to probe the high value elements first, then make decisions about lower value elements --- this gives it another advantage over the non-adaptive policy. The actual construction of the hard instance, however, builds on and generalizes the construction for prophet inequality: 

\textbf{Construction.} We again consider an $\ell$-layer tree, which, in the context of stochastic probing, we call the {\em meta-tree}. But now we replace each node (+ edges to $L^{\Theta(L^{\ell})}$ children) in the meta-tree with an $L^2$-ary  tree of depth roughly $L^{\ell}$. We call the new trees inside each layer {\em blocks} to distinguish them from the meta-tree. Together, the meta-tree with the blocks combine into one large {\em concatenated tree}. Like in the prophet inequality construction, each layer-$\ell$ element has value roughly $1/L^{\ell}$ with probability $1/L^2$, and $0$ otherwise. 

The ``inner'' (selecting) feasibility constraint is simply defined as root-leaf paths in the concatenated tree. 
The ``outer'' (probing) feasibility constraint is defined as root-leaf ``caterpillars'', i.e.~root-leaf paths + all the edges coming out of each node in the path. See also Figure~\ref{fig:trees}.

The vertices on each of the layers have to reuse elements, but we can bound the overlaps using a variant of the ``double error correcting code'' construction from prophet inequality. Specifically, the inner encoding of paths inside each block closely resembles tree codes~\cite{Schulman93,Schulman96}. Going beyond the standard path-vs-path distance guarantee of  tree codes, we use a stronger distance guarantee that bounds the intersection of the ``legs'' of one caterpillar with the ``spine'' (aka inner feasible set) of another caterpillar.

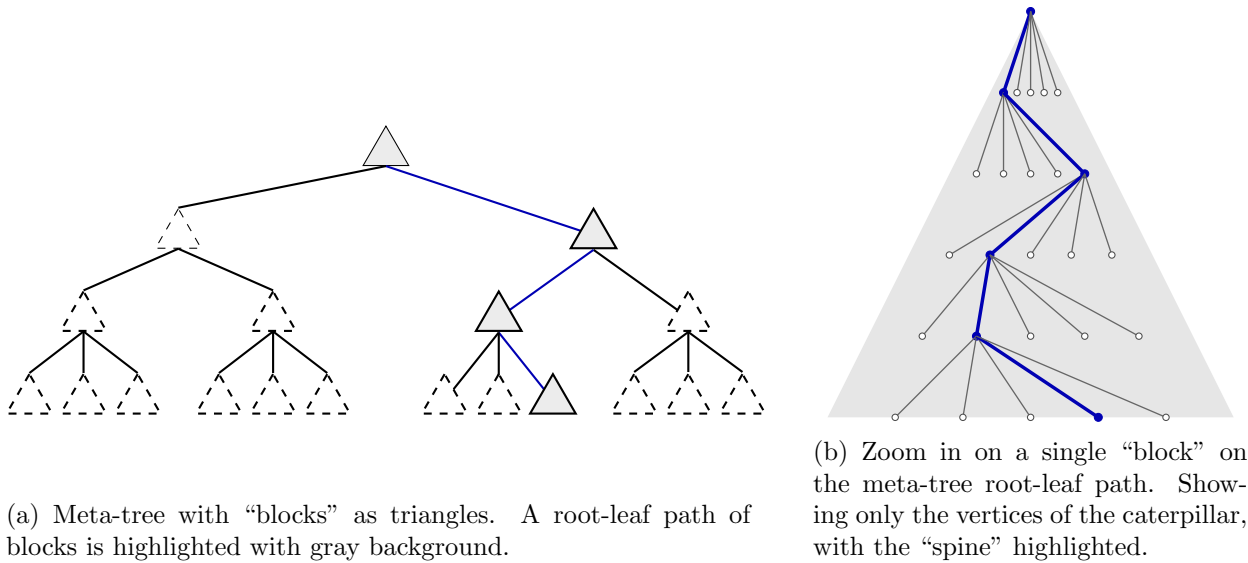
\begin{figure*}[ht]
\centering

\begin{subfigure}[b]{0.6\textwidth}
\centering
\begin{tikzpicture}[xscale = 0.6,
grow'=down,
level distance=11mm,
level 1/.style={sibling distance=92mm},
level 2/.style={sibling distance=42mm},
level 3/.style={sibling distance=12mm},
edge from parent/.style={draw, thick},
edge from parent path={(\tikzparentnode.south) -- (\tikzchildnode.north)},
tri/.style={
shape=regular polygon,
regular polygon sides=3,
% regular polygon rotate=90, % apex up (like \Delta)
draw=black, dashed,
fill=white,
minimum size=7mm, % medium-sized
inner sep=0pt
},
trip/.style={
shape=regular polygon,
regular polygon sides=3,
% regular polygon rotate=90, % apex up (like \Delta)
draw=black,
fill=gray!15,
minimum size=7mm, % medium-sized
inner sep=0pt
}
]

\node[trip] {} % root
child { [blue!70!black, very thick] node[trip] {}
child { [black, thick] node[tri] {}
child { node[tri] {} }
child { node[tri] {} }
child { node[tri] {} }
}
child { [blue!70!black, very thick] node[trip] {}
child { [blue!70!black, very thick] node[trip] {} }
child { [black, thick] node[tri] {} }
child { [black, thick] node[tri] {} }
}
}
child { node[tri] {}
child { node[tri] {}
child { node[tri] {} }
child { node[tri] {} }
child { node[tri] {} }
}
child { node[tri] {}
child { node[tri] {} }
child { node[tri] {} }
child { node[tri] {} }
}
};

\node[align=left] at (-1,-4.3) {};

%\node[align=left, anchor=west, rotate=90] at ($-(2.2,0)$) {order of arrival};

\end{tikzpicture}
\caption{Meta-tree with ``blocks'' as triangles. A root-leaf path of blocks is highlighted with gray background.}
\label{fig:tree-a}
\end{subfigure}\hfill
\begin{subfigure}[b]{0.35\textwidth}
\centering
\begin{tikzpicture}[scale=0.9,
  every node/.style={circle, draw=black, inner sep=0.6pt, minimum size=2.4pt},
  edge/.style={draw=black!60, line width=0.5pt},
  hil/.style={draw=blue!70!black, line width=1.3pt},
  sel/.style={fill=blue!70!black, draw=blue!70!black},
  normal/.style={fill=white, draw=black!70}
]

  % Parameters
  \def\H{6}     % total height of the big tree (in cm)
  \def\W{3}     % half width of the big tree at the base (in cm)
  \def\L{5}     % number of levels (root at 0, leaf at level L)
  \pgfmathsetseed{12345} % change or remove for different random paths per compile

  % Background: a big grey triangle representing the entire large tree
  \fill[gray!20] (-\W,-\H) -- (0,0) -- (\W,-\H) -- cycle;

  % Root (current node initialization)
  \def\currx{0}
  \def\curry{0}
  \node[sel, minimum size=3pt] at (\currx,\curry) {};

  % Vertical spacing between levels
  \pgfmathsetmacro{\ysep}{\H/\L}

  % Walk down the tree until reaching the bottom (a leaf at level L)
  % At each step:
  % 1) draw all children of the current node at the next level
  % 2) pick one child at random, highlight it and the edge to it
  \foreach \i in {0,...,\numexpr\L-1\relax} {

    % Random number of children (between 2 and 5)
    \pgfmathsetmacro{\c}{5}

    % y-coordinate of the next level
    \pgfmathsetmacro{\ynext}{-(\i+1)*\ysep}

    % Allowed half-width at the next level inside the triangle
    % From similar triangles: half-width = W * (level) / L
    \pgfmathsetmacro{\aw}{\W*(\i+1)/\L}

    % Horizontal spacing between children (capped so local fans don't get too wide)
    \pgfmathsetmacro{\dtemp}{2*\aw/(\c+1)}
    \pgfmathsetmacro{\d}{min(1.2,\dtemp)}

    % Compute a starting x so the c children are centered near \currx but clamped within [-aw, aw]
    \pgfmathsetmacro{\xn}{\currx - 0.5*(\c-1)*\d}
    \pgfmathsetmacro{\xminstart}{-\aw + \d}
    \pgfmathsetmacro{\xmaxstart}{\aw - \d - (\c-1)*\d}
    \pgfmathsetmacro{\xstart}{max(\xminstart, min(\xmaxstart, \xn))}

    % Randomly choose which child becomes the next current node
    \pgfmathrandominteger{\r}{1}{\c}

    % Draw all children and edges from current node
    \foreach \j in {1,...,\c} {
      \pgfmathsetmacro{\xj}{\xstart + (\j-1)*\d}

      % Normal edge
      \draw[edge] (\currx,\curry) -- (\xj,\ynext);

      % If this child is the chosen one, overdraw the edge highlighted
      \ifnum\j=\r
        \draw[hil] (\currx,\curry) -- (\xj,\ynext);
        \node[sel, minimum size=3pt] at (\xj,\ynext) {};
      \else
        \node[normal] at (\xj,\ynext) {};
      \fi
    }

    % Update current node to the chosen child (global defs so they persist across foreach iterations)
    \pgfmathparse{\xstart + (\r-1)*\d}\xdef\currx{\pgfmathresult}
    \pgfmathparse{\ynext}\xdef\curry{\pgfmathresult}
  }

\end{tikzpicture}
\caption{Zoom in on a single ``block'' on the meta-tree root-leaf path. Showing only the vertices of the caterpillar, with the ``spine'' highlighted.}
\label{fig:tree-b}
\end{subfigure}
\caption{Caterpillar construction for Stochastic Probing}
\label{fig:trees}
\end{figure*}

\textbf{Overview of Analysis.}
There is a simple greedy adaptive policy that obtains $\Omega(L)$ value (in expectation and w.h.p.): starting from the root, at each iteration probe all the children of the current vertex and proceed to one with nonzero value (such a child exists with constant probability; if none exist proceed to an arbitrary child). 

For the non-adaptive policy, each node on the main path (``spine of the caterpillar'') probed has only $1/L^2$ probability of having nonzero value. On any specific layer, the non-adaptive policy may also gain some value from the contribution of ``caterpillar legs'', however, the double error-correcting code construction guarantees that 
\begin{itemize}
    \item Inside the same block, those caterpillar legs contain at most a $1/L$-fraction of any other spine; and 
    \item When a caterpillar and a spine are in different blocks, they share at most $O(1/L^2)$-fraction of elements. Thus it can only get $O(1/L^2)$-fraction of the remaining value from any other layer. 
\end{itemize}

% \begin{itemize}
%     \item On any specific layer,  those caterpillar legs contain at most a $1/L$-fraction of any spine; and 
%     \item Once the non-adaptive policy covers more than $O(1/L^2)$-fraction of some spine with the legs of another caterpillar, its entire path to the root is fixed. Thus it can only get $O(1/L^2)$-fraction of the remaining value from other layers. 
% \end{itemize}
In sum, any non-adaptive policy is limited to at most $O(1/L)$-fraction on the best layer and $O(1/L^2)$-fraction on each of the other $L-1$ layers.

\section{Single-Log Algorithm for the Secretary Problem}  \label{sec:secretary}

We study the secretary problem for XOS functions. Formally, given a universe of elements $[n]:= \{1,\ldots,n\}$, a set function $f : 2^{[n]} \to \mathbb{R}_{\ge 0}$ is called an XOS function  if there exists a family $\calV$ of nonnegative value vectors (a.k.a. \emph{clauses}) in $\mathbb{R}_{\ge 0}^n$ such that for every $S\subseteq[n]$, 
\[
f(S) = \max_{v \in \calV} \sum_{i \in S} v_i.
\]
This class  generalizes linear and submodular functions, and is also known to $O(\log n)$-approximate any (monotone) subadditive set function \cite{dobzinski2007two}.

\begin{defn}[Secretary Problem]
We are given a downward-closed feasibility constraint $\mathcal{F}$ over elements $[n]$.  
There is an unknown XOS valuation function $f : 2^{[n]} \to \mathbb{R}_{\ge 0}$ and 
the elements arrive one by one in a uniformly random order.  At any point, the algorithm can observe values for all subsets of arrived elements (value oracle). On arrival of  element $i$, the algorithm must immediately and irrevocably decide whether to accept or reject $i$, while maintaining feasibility under $\mathcal{F}$.  
The objective is to maximize $f(S)$ for the finally accepted set $S$,  
and the benchmark is the offline optimum (hindsight) value $\max_{T \in \mathcal{F}} f(T)$.
\end{defn}

We will now design an $O(\log n)$-competitive algorithm for this problem.

\subsection{Preliminaries and Preprocessing}
 By w.h.p.\,we mean $1-1/n^c$ where the constant $c$ can be made arbitrarily large by increasing the constant in $O(\log n)$-competitive algorithm.

\paragraph{Converting to XOS Representation.}
Although we are given value oracle access to the XOS function $f$ on the arrived elements, it will be useful for our proofs to think of its XOS representation. Formally,  a vector $v \in \R^{n}_{\geq 0}$ is said to \emph{support}  set $S \subseteq [n]$ if for every subset $T\subseteq S$, we have $\sum_{i \in T} v_i \leq f(T)$. Thus, to obtain XOS representation of $f$ online, we add all supporting vectors on the set of already arrived elements as XOS clauses into the clause family $\calV$.

Observe that this means the clauses will expand as more elements arrive, and that the clauses are ``downward-closed'': if vector $a\in \R^{n}_{\geq 0}$ is a clause and vector  $b\in \R^{n}_{\geq 0}$ is coordinate-wise less than $a$, then vector $b$ is also a clause. Although this means that there are infinitely many clauses, it suffices to work with vector-coordinates that are powers of 2, which we discuss next.

\paragraph{Standard Preprocessing.}
W.l.o.g.\,we consider instances where: (i) all values in the XOS clauses $\calV$ are powers of $2$; (ii) all values in the XOS clauses are in the range $[1/n^2,1]$; and (iii) every individual element contributes at most $1/\omega(\log n )$-fraction of the optimal value (equivalently, of the clause achieving the optimal value); this implies that the top $\log(n)$ elements contribute $o(1)$-fraction of the total value.
%Also, wlog we only focus on value classes that have at least $|W_C| = \Omega(\log(n))$ elements.
To see why all these assumptions are w.l.o.g., consider the following preprocessing:
\begin{itemize}
    \item To get an instance with powers of $2$, we round down every element's contribution in each clause to the nearest power of $2$, which  decreases the total value of any set of elements by at most a factor of $2$.
   
    \item With probability $1/2$ we run the single-choice secretary algorithm. Thus, we can henceforth assume w.l.o.g.\,that 
    each element $1/\omega(\log n )$-fraction of the total value;  otherwise, the single-choice secretary algorithm already guarantees an $\Omega(1/\log(n))$-fraction. 
    
    \item If we don't run the single-choice algorithm, we use the first half of the stream as a {\em pre-sample}. With high probability, one of those top $\log(n)$ elements appears in this pre-sample and the remaining second half will  also have a constant fraction of the total optimal value. We can use the top value from this pre-sample $a_{*}$ to normalize every element's contribution in each clause:
    $$
    \hat{a}_i := 
    \begin{cases}
    0 & \text{ if }~ a_i < a_{*}/n^2 \\
    a_i/ a_{*} & \text{ if }~ a_i \in [a_{*}/n^2, a_{*}] \\
    1 & \text{ if }~ a_i > a_{*}
    \end{cases}.
    $$
    Note that the total value of elements rounded down to $0$ is negligible compared even to just $a_{*}$.
    Furthermore, w.h.p.\,we only round down to $1$ the top fewer-than-$\log(n)$ elements, so this truncation has negligible impact on the total value. 
    %\item The total contribution of all elements from value 
\end{itemize}
Henceforth, for ease of presentation, ignore these preprocessing steps and simply assume that we are given an instance in XOS representation satisfying our assumptions. This also implies that the optimal value of this instance is $\Omega(\log n)$, since we know that a single element can already get $a_{*}$ value and item (iii)  of preprocessing implies $a_{*}$ is less than $(1/\log n)$ fraction of the optimum.

\subsection{Setting up the Algorithm}

\paragraph{Tie-breaking.}
Even though we consider element values in each clause to be a power of $2$,  we will assume for convenience of notation that each value is perturbed by an infinitesimal amount. We will not account for the contribution of these perturbations, but they allow us to assume w.l.o.g. that optimum subsets are unique.

%Also, wlog we only focus on value classes that have at least $\Omega(\log(n))$ elements (the total value is at least $\omega(\log(n))$ - otherwise single-item secretary guarantees approximation factor; and the remaining value classes have negligible total contribution),
%\Aviad{TBD: formalize}
\paragraph{Lower-Sample, Current-Sample, and Real elements.}
Assuming\footnote{This is wlog. If $n$ is not a multiple of $3$, we can add one or two dummy elements with value $0$ in every clause.} $n$ is a multiple of $3$, we partition the stream into $3$ parts:
\begin{itemize}
    \item The first third $I_1 := \{1,\dots,n/3\}$ are the {\em Current-Sample};
    \item the middle\footnote{The order between Lower-Sample is not important.} third $I_2 := \{n/3+1,\dots,2n/3\}$ is the {\em Lower-Sample}; and
    \item the last third $I_3 := \{2n/3+1,\dots,n\}$ is the {\em Real}.
\end{itemize}
Intuitively, for each scale $C$, our algorithm will try to select a subset of Real elements with marginal contribution $C$ that is simultaneously compatible with (i) all previously selected elements with marginal contribution $>C$; (ii) a subset of Current-Sample elements with marginal contribution $C$; and (iii) a subset of Lower-Sample elements with marginal contribution $<C$.

Our analysis exploits the symmetry, due to random order arrival, between the three parts. Specifically, it will be useful to compare the $t$-th element in the Current-Sample with the $t$-th element in the Real (the order among Lower-Sample elements is not important). To make notation convenient, for element $t \in [n/3]$ in the Current-Sample, we use $\ot$ to denote the corresponding Real element $t+2n/3$.
We further define $I_1(t) := \{t+1,\dots,n/3\}$, and similarly $I_3(t) := \{t+2n/3+1,\dots,n\}$.

After the algorithm has already selected set $A$ with total value $\tau_A$, we will need two notions of ``optimal'' solutions: (i)~the first one maximizes total value at all scales $\leq C$, and (ii)~the second one only maximizes  total value (number of elements) at scale $C$ while respecting a total lower value constraint. We formally define both of these next.

\paragraph{(i) $\Opt^{\leq C}(I \mid  A,\tau_A)$ notation.}
Given sets $I$, $A$, scale $C$, and a threshold $\tau_A$, we consider the set $S$ and corresponding XOS-clause $v_{\cdot,j^*}$ that maximize $\sum_{t \in S} v_{t,j^*}$ subject to
\begin{enumerate}
    \item $S$ is feasible;
    \item $S \supseteq A$ and $\sum_{t \in A} v_{t,j^*} \ge \tau_A$ ; 
    %\item for all $t \in I_{\text{Current}}$, we have $v_{t,j^*} = C$; 
    \item for all $t \in I \cap S$, we have $v_{t,j^*} \leq C$.
\end{enumerate}
With slight abuse of notation we let $\Opt^{\leq C}(I \mid  A,\tau_A)$ denote both this optimum set and its value.

\paragraph{(ii) $\Opt^C(I_{\text{Current}} \mid  A,\tau_A \mid I_{\text{Lower}},\tau_{\text{Lower}})$ notation.}
Given sets $I_{\text{Current}},I_{\text{Lower}},A$, scale $C$, and thresholds $\tau_{\text{Lower}}, \tau_A$, we consider the (unique, due to tie-breaking) set $S \subseteq I_{\text{Current}}$ and corresponding XOS-clause $v_{\cdot,j^*}$
that maximize $\sum_{t \in S} v_{t,j^*}$ such that there exists subset $T \subseteq I_{\text{Lower}}$ satisfying
\begin{enumerate}
    \item $A \cup S\cup T$ is feasible;
    \item $\sum_{t \in A} v_{t,j^*} \ge \tau_A$ ; 
    \item for all $t \in S$, we have $v_{t,j^*} = C$; 
    \item for all $t \in T$, we have $v_{t,j^*} \leq  C/2$ and  $\sum_{t \in T} v_{t,j^*} \ge \tau_{\text{Lower}}$.
\end{enumerate}
Note that since all elements of $S$ have clause-value exactly 
$C$, the objective is really maximizing the cardinality of
$S$ (up to infinitesimal tie-breaking), while $T$ only serves as a lower-scale witness.
With slight abuse of notation we let $\Opt^C(I_{\text{Current}} \mid  A,\tau_A \mid I_{\text{Lower}},\tau_{\text{Lower}})$ denote both this optimum set and its value.

\paragraph{Labels of Real elements.}
We further partition the Real elements with {\em labels} that assign them to different scales.
Ideally, we want to think of  assigning to each Real element $\ot$ a uniformly random label $C(\ot)$ independently and uniformly at random. 
Intuitively, we will try to cover the contribution from value-$C$ elements with Real elements with label $C(\ot)$.

In order to respect the actual (uniformly random) order of arrival, we will assign labels to the {\em indices} of Real element. Whenever we want to access a Real element corresponding to some label and index $\ot$, we will reveal the next Real element in the actual order of arrival.

\subsection{Algorithm}
We begin with a description of a simpler, ideal algorithm that doesn't exactly respect the true arrival order (Algorithm~\ref{Alg:analysis}), and then show an equivalent implementable variant (Algorithm~\ref{Alg:secretary}). Our ideal algorithm proceeds in $O(\log(n))$ phases, one for each scale $C$. At the beginning of each phase for scale $C$, it initializes the phase by  computing a current offline optimal solution among sample elements (specifically $\Opt^{\leq C}(I_2 \mid  \Alg,\tau_{\Alg})$ as defined above) to set the reserve threshold for expected contribution from lower scales.

The main subroutine for each scale proceeds by incrementing the index $t$ for the Current-Sample and $\ot$ for the Real. For each real element $\ot$ it asks if (i) its label $C(\ot)$ matches the current scale; and (ii) it would be in $\Opt^C(\ot \cup I_1(t) \mid \Alg,\tau_{\Alg} \mid I_2, \tau_{<C})$, i.e.~the optimal set at scale $C$ among remaining Current-Sample elements $I_1(t)$, while ensuring compatibility with previously selected elements $\Alg$ and future lower-scale elements as predicted by $I_2$. If $\ot$ satisfies both conditions, it is added to $\Alg$ and we increase $\tau_{\Alg}$ by $C$.

\begin{algorithm}[t] 
\caption{Ideal Secretary algorithm}
\label{Alg:analysis}
\textbf{Input:} Secretary problem instance satisfying properties (i), (ii), and (iii) of standard preprocessing.\medskip
\begin{enumerate}
\item  $\Alg \leftarrow \emptyset$ and $\tau_{\Alg} \leftarrow 0$.
\item for each $C\in \{1,2^{-1},2^{-2}, \dots, \underbrace{2^{-\lceil 2\log(n) \rceil}}_{\approx 1/n^2}\}$: 
    \begin{enumerate}
        \item Compute $\Opt^{\leq C}(I_2 \mid  \Alg,\tau_{\Alg})$, and let $\tau_{<C}$ denote the total value from scales $\leq C/2$.
        %$\Opt^C(I_1 \mid \Alg,\tau_{\Alg} \mid I_2,0)$, and let $\tau_{<C}$ denote the total contributions of $I_2$-elements to $\Opt^C(I_1,I_2 \mid \Alg)$.
        \item For each $t \in I_1$ (in order), 
        {\color{blue} \begin{itemize}
            \item If $C(\ot) = C$ AND $\ot \in \Opt^C \big(\ot \cup I_1(t)\mid \Alg,\tau_{\Alg} \mid I_2, \tau_{<C} \big)$,
            %$\ot \in \Opt^C(\ot \cup I_1(t) \mid \Alg,\tau_{\Alg} \mid I_2, \tau_{<C})$,
            \begin{enumerate}
                \item $\Alg \leftarrow \Alg \cup \{\ot\}$.
                \item $\tau_{\Alg} \leftarrow \tau_{\Alg}+ C$
            \end{enumerate}
        \end{itemize}}
    \end{enumerate}    
\item Return $\Alg$

\end{enumerate}
\end{algorithm}

The ideal algorithm restarts the Real part of the stream for each phase, which doesn't respect the actual order of arrival. However, each Real element is only accessed in one, uniformly random phase (depending on its label). Thus, due to the random order of the stream, our ideal algorithm is equivalent to the implementable variant (Algorithm~\ref{Alg:secretary}) that pops the next Real element each time it sees an index corresponding to the current label. 

\begin{algorithm}[t] 
\caption{Implementable Secretary algorithm}
\label{Alg:secretary}
\textbf{Input:} Secretary problem instance satisfying properties (i), (ii), and (iii) of standard preprocessing.\medskip
\begin{enumerate}
\item  $\Alg \leftarrow \emptyset$ and $\tau_{\Alg} \leftarrow 0$.
\item $t_\text{Real} \leftarrow \overline{0} = 2n/3$ {\color{gray-comment} \# index for actual real elements}
\item for each $C\in \{1,2^{-1},2^{-2}, \dots, \underbrace{2^{-\lceil 2\log(n) \rceil}}_{\approx 1/n^2}\}$: 
    \begin{enumerate}
    \item Compute $\Opt^{\leq C}(I_2 \mid  \Alg,\tau_{\Alg})$, and let $\tau_{<C}$ denote the total value from scales $\leq C/2$.
        %\item Compute $\Opt^C(I_1 \mid \Alg,\tau_{\Alg} \mid I_2,0)$, and let $\tau_{<C}$ denote the total contributions of $I_2$-elements to $\Opt^C(I_1,I_2 \mid \Alg)$.
        \item For each $t \in I_1$ (in order), 
        {\color{blue} \begin{itemize}
            \item If $C(\ot) = C$
            \begin{enumerate}
                \item $t_\text{Real} \leftarrow t_\text{Real}+1$
                \item If $t_\text{Real} \in \Opt^C(t_\text{Real} \cup I_1(t) \mid \Alg,\tau_{\Alg} \mid I_2, \tau_{<C})$,
                \begin{enumerate}
                    \item $\Alg \leftarrow \Alg \cup \{t_\text{Real}\}$
                    \item $\tau_{\Alg} \leftarrow \tau_{\Alg}+ C$
                \end{enumerate}              
            \end{enumerate}
        \end{itemize}}
    \end{enumerate}    
\item Return $\Alg$

\end{enumerate}
\end{algorithm}

\begin{observation} \label{obs:algoEquiv}
    Under a uniformly random order of arrival, Algorithms~\ref{Alg:analysis} and~\ref{Alg:secretary} are equivalent. 
\end{observation}

Throughout the analysis, we consider the simpler, ideal algorithm.

\subsection{Analysis}

For $t \in [n/3]$ and scale $C$, we use $\Alg^{(C,t)},\tau_{\Alg}^{(C,t)}$ to denote the value of these parameters after processing $\ot$ in phase-$C$ of the algorithm. We abuse notation and use $(C,0)$ for the beginning of the phase.

Let $\Opt_C$ denote the set of $C$-value elements in  the optimal set at the beginning of phase $C$ of the algorithm, i.e., in $\Opt^{\leq C}(I_2 \mid \Alg,\tau_{\Alg})$. 
%$\Opt^C(I_1 \mid \Alg^{(C,0)},\tau_{\Alg}^{(C,0)} \mid I_2,0)$ denote the number of $C$-value elements (aka $I_1$-elements) in the optimal set at the beginning of phase $C$ of the algorithm. 
Let $\Alg_C :=  \Alg^{(C,n/3)} \setminus  \Alg^{(C,0)}   $ denote set of elements added to $\Alg$ during the $C$-phase. Our goal is to show that, as long as $|\Opt_C|$ isn't negligible, then w.h.p.\ we have $|\Alg_C| = \Omega(|\Opt_C|/\log(n))$.

\begin{lemma}\label{lemma:mainXOS}
    If $|\Opt_C| > 100\log(n)\cdot (|\Alg^{(C,0)}|+1)$, then whp $|\Alg_C| = \Omega(|\Opt_C|/\log(n))$.
\end{lemma}
We defer the proof of this lemma, which is the heart of the proof, to  \Cref{sec:lemmaSingleScale}. 
We also need the following easy auxiliary claim on the concentration of the optimum among sample elements to  prove our main theorem.

\begin{claim}%[Value of $\Opt(\text{Sample})$ concentrates]
\label{claim:opt(sample)}
W.h.p., $\Opt([n/3]) \ge  \Opt([n])/4$.
\end{claim}

\begin{proof}
    Let $v$ denote the coefficients of the XOS-maximizing clause for $\Opt([n])$; we denote $\Opt^{[n/3]}([n]) = \sum_{t \in \Opt([n]) \cap [n/3]} v_t $.  Observe that $\Opt([n/3]) \ge \Opt^{[n/3]}([n])$. In other words, we have a lower bound of $\Opt^{[n/3]}([n])$, aka the contribution of the Sample elements to the global optimum. We can bound the probability of a large deviation from the expected $\Opt([n])/3$ using, e.g., Hoeffding's Inequality: 
    \begin{align*}
        \Pr\left[\Opt^{[n/3]}([n])  <  \Opt([n])/4\right] & < \exp\left(-\Omega \left(\frac{\Opt([n])^2}{\sum_{t \in \Opt([n])}v_t^2} \right)\right) \\
        &= \exp(-\omega(\log n )) \enspace,
    \end{align*}
    where $ \frac{\Opt([n])^2}{\sum_{t \in \Opt([n])}v_t^2} \geq \frac{\sum_{t \in \Opt([n])}v_t}{\max_{t \in \Opt([n])}v_t}  = \omega(\log n )$ follows by item (iii) of standard preprocessing.
\end{proof}

Now we can prove our main result.

\begin{theorem} \label{thm:mainSecretFormal}
 \Cref{Alg:secretary} is  $O(\log n )$-competitive for the secretary problem with an XOS objective and general downward-closed feasibility constraints.
\end{theorem}

\begin{proof}
    By \Cref{obs:algoEquiv}, it suffices to analyze \Cref{Alg:analysis}.
    
    For \Cref{Alg:analysis}, we will show the invariant that  w.h.p. at the start of each phase $C$,
    \begin{align}\label{eq:invariant}
     \tau_{\Alg}  + \frac{1}{1000 \log(n)} \Opt^{\leq C}(I_2 \mid  \Alg,\tau_{\Alg})  =  \Omega\Big(\frac{\Opt([n])}{\log(n)} \Big).    
    \end{align}
    At the beginning of the algorithm (i.e., $C=1$), we have $\tau_{\Alg} = 0$ and $ \Opt^{\leq C}(I_2 \mid  \Alg,\tau_{\Alg}) = \Opt([n/3])$, so the invariant holds w.h.p.\,by \Cref{claim:opt(sample)}. We will prove below by induction that it holds for every $C$. 
    
    To see why this invariant suffices, observe that 
    at the end of the algorithm (i.e., $C=1/(2n^2)$) the second term $\Opt^{\leq C}(I_2 \mid  \Alg,\tau_{\Alg})  = 0$ as there are no elements of smaller value in the clauses by item (ii) of preprocessing. Since, $f(\Alg) \geq \tau_{\Alg}$,  the theorem follows.

    \paragraph{Proving the Invariant \eqref{eq:invariant}.} Given that the invariant holds at value scale $C$ (starting with $C=1$), consider changing the scale from $C$ to  $C/2$, i.e., going from beginning of phase $C$ to beginning of phase $C/2$. 
    Since the lower bound constraint $\tau_{<C}$ ensures that the total value from values $<C$ from $I_2$ is protected as $\Alg$ changes during phase $C$, the drop in $\Opt^{\leq C}(I_2 \mid  \Alg,\tau_{\Alg})$ on the l.h.s.\ is at most $C \cdot |\Opt_C|$. Consider two cases depending on the size of $\Opt_C$:

    \begin{enumerate}
        \item    
    If $|\Opt_C| < 100 \log(n) \cdot (|\Alg^{(C,0)}|+1)$ at the beginning of phase $C$, we can simply discard the contribution of these  elements: For every value-$C$ element selected by our algorithm, we are discarding potential contributions from up to $100 \log(n)$ elements from each smaller value, for a total value of up to $200\log(n)\cdot C$ (as values decrease in powers of $2$); even if we had successfully collected $1/\log(n)$-fraction of this value, we would only increase the total value of our algorithm by a constant factor.  
    
    \item Otherwise, i.e.~when $|\Opt_C| \geq 100 \log(n) \cdot (|\Alg^{(C,0)}|+1)$  at the beginning of phase $C$, \Cref{lemma:mainXOS} guarantees that w.h.p.\,we collect $|\Alg_C| =    \Omega(|\Opt_C|/\log (n))$ elements of value-$C$ during the $C$-phase, each of whom increase $\tau_{\Alg}$ by $C$ in the first term, thereby maintaining the invariant. 
    \qedhere
    \end{enumerate}
\end{proof}

\subsection{Proof of the main lemma (\Cref{lemma:mainXOS})} \label{sec:lemmaSingleScale}
We want to prove $|\Alg_C|=  |\Alg^{(C,n/3)} \setminus  \Alg^{(C,0)}| = \Omega(|\Opt_C|/\log n) $, where $\Opt_C$ denotes the $C$-value elements in  the optimal set at the beginning of phase $C$ (i.e., in $\Opt^{\leq C}(I_2 \mid \Alg,\tau_{\Alg})$).   Thus, at the beginning of phase $C$, we have $C \cdot \Opt_C = \Opt^C \big(I_2 \mid \Alg,\tau_{\Alg} \mid I_2, \tau_{<C} \big)$

We prove this lemma in two steps:
\begin{enumerate}
    \item First, we prove that during any fixed phase $C$, the number of elements  $|\Alg_C|$ selected by the algorithm satisfies,   w.h.p.,  
    \begin{align} \label{eq:singleScale}
        |\Alg_C| = \Omega\Big(\Opt^C \big(I_1 \mid \Alg,\tau_{\Alg} \mid I_2, \tau_{<C} \big)\Big) \cdot 1/(C\log n).
    \end{align}
    The proof of this lemma relies on defining good and bad events, similar to Rubinstein \cite{Rubinstein16}. The reason we lose $\log n$ factor is because only $1/\log n$ fraction of times are labeled $C$.
    
    \item Second, we prove a  ``transfer lemma'', which  at the change of phase (say from $C$ to $C/2$) allows us to transfer protected value (due to $\tau_{<C}$) at scale $C/2$ in $I_2$ to scale $C/2$ in $I_1$. Formally, we show that, w.h.p.,
    \begin{align}\label{eq:transferLemma}
        \Opt^C \big(I_1 \mid \Alg,\tau_{\Alg} \mid I_2, \tau_{<C} \big)  = \Omega\Big(\Opt^C \big(I_2 \mid \Alg,\tau_{\Alg} \mid I_2, \tau_{<C} \big) \Big) .
    \end{align} 
   
\end{enumerate}

\subsubsection{Step 1: Analyzing a single scale $C$}
\begin{proof}[Proof of \Cref{eq:singleScale}]
 Since the inequality is difficult to prove directly  w.r.t.\ $\Alg$, $\tau_{\Alg}$, and $\tau_{<C}$, we first prove that the inequality fails with an exponentially small probability for any given set $A$ and thresholds $\tau_A$, $\tau_{<C}$, and then we take a union bound over every possible choice of $A, \tau_A, \tau_{<C}$.

    Consider some fixed choice of thresholds $\widehat{\tau_{<C}},\widehat{\tau_{A}}$ and set $\widehat{A}$.
    Following~\cite{Rubinstein16}, we define the following good and bad events for each $t$:
    \begin{itemize}
        \item Good event $G(t) = G(t,\widehat{\tau_{<C}},\widehat{\tau_{A}},\widehat{A})$: $\ot \in \Opt^C(\ot \cup I_1(t) \mid \widehat{A},\widehat{\tau_{A}} \mid I_2, \widehat{\tau_{<C}})$.
        \item Bad event $B(t) = B(t,\widehat{\tau_{<C}},\widehat{\tau_{A}},\widehat{A})$: $t \in \Opt^C(t \cup I_1(t) \mid \widehat{A},\widehat{\tau_{A}} \mid I_2, \widehat{\tau_{<C}})$
    \end{itemize}
    The number of $I_1(t)$-elements in the set achieving $\Opt^C(I_1(t) \mid \Alg,\tau_{\Alg} \mid I_2, \tau_{<C})$ decreases by at most $1$ with each bad event (this is the reason we use $\Opt^C$ within a phase since it maximizes the value at scale $C$, instead of total value as in $\Opt^{\leq C}$), and it doesn't change at all otherwise. Since at the end of the phase we have no $I_1(t)$-elements left and each bad event decreases value by at most $C$, we have $C\cdot \sum_t B(t) \ge \Opt^C(I_1(t) \mid \Alg,\tau_{\Alg} \mid I_2, \tau_{<C})$. Our goal now is to show that, w.h.p., the number of good events can be lower bounded using the number of bad events (irrespective of the labels). We later show that $|\Alg_C|$ can be lower bounded in terms of the number of good events, where we will use the randomness in labels.

    \paragraph{Many good events.}
    Events $G(t)$ and $B(t)$ are equally likely by symmetry between $t$ and $\ot$. Furthermore, they remain equally likely even if we condition on any suffix of $I_1$ and $\overline{I_1}$. Therefore, if we reveal the identity of $t$ vs $\ot$ in reverse random order, $\sum_t \big(G(t) - B(t)\big)$ forms a martingale. 
    To bound the deviation of this martingale, we will use the following martingale version of Bennett's inequality%
\footnote{The main advantage over the more popular Azuma's inequality is that
this theorem depends on the variance rather than an absolute bound
on the values each variable can take; see also~\cite{Rubinstein16}.}.
\begin{thm}[e.g., \cite{Freedman1975}]
\label{thm:martingale-concentration} Let
$\left(\Omega,{\cal H},\mu\right)$ be a probability triple, with
${\cal H}_{0}\subseteq{\cal H}_{1}\subseteq\dots$ an increasing sequence
of sub-$\sigma$-fields of ${\cal H}$. For each $k=1,2,\dots$ let
$X_{k}$ be an ${\cal H}_{k}$-measurable random variable satisfying
$\left|X_{k}\right|\le1$ and $\E\left[X_{k}\mid{\cal H}_{k-1}\right]=0$.
Let $S_{n}\triangleq\sum_{k=1}^{n}X_{k}$ and $T_{n}\triangleq\sum_{k=1}^{n}\E\left[X_{k}^{2}\mid{\cal H}_{k-1}\right]$.
Then, for every $a,b>0$,
\[
\Pr\left[\left(S_{n}\geq a\right)\wedge\left(T_{n}\leq b\right)\right] ~\leq~ \exp\left(\frac{-a^{2}}{2\left(a+b\right)}\right).
\]
\end{thm}

For each $t$, the variance of the martingale is $(G(t) - B(t))^2 \le G(t) + B(t)$, therefore the total variance $T_n$ is upper bounded by $\sum_t B(t)$.
Therefore, if we set $a = \sum_t B(t) /2$ and $b = \sum_t\big(G(t) + B(t)\big)$, we have that $\sum_t G(t) \ge \sum_t B(t) /2$ except with probability  
\begin{align}\label{eq:G=B}
    \exp \left(- \frac{\big(\sum_t B(t) \big)^2}{2\sum_t \big(G(t) + B(t)\big)}\right) ~\le~ \exp(-\Opt_C/4).     
\end{align}

Now, we can take a union bound over all choices of $a,b,\widehat{\tau_{<C}},\widehat{\tau_{A}},\widehat{A}$. In particular the number of choices in the union bound is dominated by the ${{O(n)}\choose{|\widehat{A}|}}$ choices of $\widehat{A}$. Restricting to\footnote{Technically we cannot use the knowledge of $\Alg^{(C,0)}|$,$|\Alg_C|$ when setting up the concentration inequality, but we can take the union bound also over this parameter.} $|\widehat{A}| \le |\Alg^{C,n/3}| =  |\Alg^{(C,0)}| + |\Alg_C| < |\Opt_C|/10\log(n)$,  we have $\leq n^{|\Opt_C|/10\log(n)}$ choices overall, which is negligible compared to the concentration we get from~\eqref{eq:G=B}.

\paragraph{The labels of good events.}
We have shown that, w.h.p.\ over the $I_1$-vs-$\overline{I_1}$ ordering of elements, we have many good events for any choice of $\Alg$: 
\[\sum_t G(t) \ge 1/2\cdot \sum_t B(t) \ge 1/(2C)\cdot \Opt^C(I_1(t) \mid \Alg,\tau_{\Alg} \mid I_2, \tau_{<C}).\] 
Now, fix such a choice of $I_1$-vs-$\overline{I_1}$, and reveal the labels of $\ot$ elements in forward order. Each good event $G(t)$ converts to an element added to $\Alg_C$ iff $C(\ot) = C$, i.e., with probability $1/\log(n)$, even conditioned on any choice of labels for earlier $\ot$. Therefore, $\sum_t\big( G(t) - \log n \cdot |\Alg_C^{(t)}| \big)$  forms a martingale, where we use $\Alg_C^{(t)}$ to denote the set of $C$-phase elements added to $\Alg$ by the time we finish processing $\ot$.  Using the same martingale concentration inequality (Theorem~\ref{thm:martingale-concentration})%
\footnote{Note that now the variance is larger, namely $O(\log(n))$ for each good event. But because we don't need to take another union bound over $n^{O(|\Opt_C|/\log(n))}$ choices of $\widehat{A}$ the weaker tail bound of $\exp(-\Omega(\Opt_C/\log(n)))$ is sufficient.}, we also have that w.h.p.\ over the choice of labels, at the end of the $C$-phase, 
\[
|\Alg_C|
=
\Omega\!\left(\frac{\sum_t G(t)}{\log n}\right)
=
\Omega\!\left(\frac{\Opt^C(I_1(t) \mid \Alg,\tau_{\Alg} \mid I_2, \tau_{<C})}{C\log n}\right).
\]
which completes the proof of  \Cref{eq:singleScale}.
\end{proof}

\subsubsection{Step 2: Transition between scales}
\begin{proof}[Proof of \Cref{eq:transferLemma}]
Similar to the proof of \Cref{eq:singleScale},  we first prove that the inequality fails with an exponentially small probability for any given set $A$ and thresholds $\tau_A$, $\tau_{<C}$, and then we take a union bound over every possible choice of $A, \tau_A, \tau_{<C}$.

Let us fix some choice of $\widehat{\Alg},\widehat{\tau_{\Alg}}, \widehat{\tau_{<C}}$. To simplify the proof, we analyze a slightly different partitioning between $I_1$ and $I_2$, where we first condition on their union $I_1 \cup I_2$ and then each element is placed in either interval i.i.d.\   with probability $1/2$  (instead of a uniform partitioning where we always have $|I_1|=|I_2|$). We will be prove that for this i.i.d.\ partitioning the inequality fails with an exponentially small probability. Since i.i.d.\ partitioning has at least $\Omega(1/\sqrt{n})$ of having  $|I_1|=|I_2|$, this would imply that even for uniform partitioning the probability of the inequality failing is exponentially small. Henceforth, consider i.i.d.\ partitioning.

%Due to the random order, this is exactly equivalent to using contiguous intervals, where we sample the length of $I_1$ from $\text{Binomial}(2n/3,1/2)$. It is also approximately equivalent to the true partition at exactly $n/3$ --- with high probability and up to $1\pm \tilde{O}(1/\sqrt{n})$-factor. 

% Our goal is thus to prove that whp:

% \begin{align}\label{eq:transferLemma-hats}
%         \Opt^C \big(I_1 \mid \widehat{\Alg},\widehat{\tau_{\Alg}} \mid I_2, \widehat{\tau_{<C}} \big)  = \Omega\Big(\Opt^C \big(I_2 \mid \widehat{\Alg},\widehat{\tau_{\Alg}} \mid I_2, \widehat{\tau_{<C}} \big) \Big).
%     \end{align} 
Consider the set $\widehat{\Opt_C} := \Opt^C \big(I_1 \cup I_2 \mid \widehat{\Alg},\widehat{\tau_{\Alg}} \mid I_2, \widehat{\tau_{<C}} \big)$ that maximizes the {\em total} $C$-elements contribution from both intervals, subject to sufficient $<C$-contributions from $I_2$ (and the usual constraints of feasibility and $\widehat{\Alg},\widehat{\tau_{\Alg}}$-contribution). Let $T_{<C} \subseteq I_2$ be the maximum-value set supporting the $<C$ contribution towards $\Opt^C \big(I_1 \cup I_2 \mid \widehat{\Alg},\widehat{\tau_{\Alg}} \mid I_2, \widehat{\tau_{<C}} \big)$.

Now, resample the partition $I_1,I_2$, conditioning on $T_{<C} \subseteq I_2$ being the maximum value set supporting the $<C$ contribution towards $\widehat{\Opt_C}$.
For this fresh resampling of the partition, reveal the $I_1$-vs-$I_2$ labels of elements in $\Opt^C$ sequentially. 
%For each one, even after we condition on all $I_1$-vs-$I_2$ labels of previous elements, the conditioning on no $<C$-subset of $I_2$ having greater  value than $T_{<C}$ (in any XOS clause), means each element is only more likely to be in $I_1$. Therefore letting $R$ be the set of elements whose $I_1$-vs-$I_2$ labels have been revealed, $|R \cap \widehat{\Opt_C} \cap I_1| - |R \cap \widehat{\Opt_C} \cap I_2|$ is a sub-martingale as $R$ increases. 
To justify the submartingale formally, fix a deterministic tie-breaking order on all clause--witness
pairs \((a,T)\), ordered first by the corresponding global \(C\)-rank on \(I_1\cup I_2\) and then
lexicographically. Let \((a^\star,T^\star)\) be the first available pair, i.e. the first pair with
\(T^\star\subseteq I_2\). Conditioning on \((a^\star,T^\star)\), after fixing \(T^\star\subseteq I_2\),
the remaining \(I_1/I_2\)-labels are still independent fair coins, and the additional condition that
no earlier pair is available is a decreasing event in the \(I_2\)-coordinates. Hence, by FKG,
each still-unrevealed element of \(\widehat{\Opt}_C\) is conditionally at most as likely to lie in
\(I_2\) as in \(I_1\). Therefore the process
$|R\cap \widehat{\Opt}_C\cap I_1|-|R\cap \widehat{\Opt}_C\cap I_2|$
is a submartingale as the labels are exposed.

We can therefore apply the same martingale concentration bound (\Cref{thm:martingale-concentration}) to argue that once we finish revealing all the labels, $|I_1 \cap \widehat{\Opt_C}| \ge |I_2 \cap \widehat{\Opt_C}|/2$ except with probability $2^{-\Omega(\widehat{\Opt_C})}$.
Since $\widehat{\Opt_C}|/2 \geq |\widehat{\Opt_C}|/3$ also holds (except with probability $2^{-\Omega(\widehat{\Opt_C})}$), we have 
\begin{gather}\label{eq:I_1-vs-I1cupI2}
    |I_1 \cap \widehat{\Opt_C}| \ge |\widehat{\Opt_C}|/3,
\end{gather}
except with probability $2^{-\Omega(\widehat{\Opt_C})}$. By union bound, this inequality holds w.h.p.\ over all choices of $\widehat{\Alg},\widehat{\tau_{\Alg}}, \widehat{\tau_{<C}}$ with $|\widehat{\Alg}| = \frac{\widehat{\Opt_C} \cap I_2}{10 \log(n)}$.

We therefore have w.h.p., $\Opt^C \big(I_1 \mid \Alg,\tau_{\Alg} \mid I_2, \tau_{<C} \big)$ is
\begin{align*}
     & \ge |I_1 \cap \Opt^C \big(I_1 \cup I_2 \mid \Alg,\tau_{\Alg} \mid I_2, \tau_{<C} \big)|  \\
    & \ge \Opt^C \big(I_1 \cup I_2 \mid \Alg,\tau_{\Alg} \mid I_2, \tau_{<C} \big)/3  \\
    & \ge \Opt^C \big( I_2 \mid \Alg,\tau_{\Alg} \mid I_2, \tau_{<C} \big)/3,
\end{align*}    
which completes the proof of \Cref{eq:transferLemma}.
\end{proof}

\section{Brief (and Optional) Preliminaries on Error-Correcting Codes}
This section introduces basic concepts from the study of error-correcting codes (ECC) that inspire the constructions of hard instances for Prophet Inequality (Section~\ref{sec:prophet}) and Stochastic Probing (Section~\ref{sec:probing}). We note that this section is optional: our constructions in upcoming sections do not rely on any particular results from coding theory, and we analyze them using self-contained elementary probabilistic arguments. But we include this brief section for better context and intuition. 

An {\em error correcting code} is a family of {\em codewords}, or vectors (usually of a fixed length and over a fixed field or alphabet). The two key properties of an error correcting code are:
\begin{itemize}
    \item {\bf Distance:} ECCs usually come with the guarantee that any two codewords should be at least some distance apart. In the original context of communication, if a codeword-message is corrupted by some errors, we can decode the corrupted message to the nearest codeword, which is hopefully also the original codeword-message. In this paper we consider relative Hamming distance, i.e., we consider Hamming distance between codewords normalized by the length of the codewords.
    \item {\bf Rate:} ECCs rely on redundant encodings to guarantee distance, and the {\em rate} is a measure of how efficient the encoding, or how large the family or codewords is. It is usually measured as $\frac{\log(\text{\# of codewords})}{\text{length of each codeword}}$; the details are not important for our paper.
\end{itemize}

\begin{example}
    The code $\{ABCD, GFED\}$ over alphabet $\{A,B,C,D,E,F,G\}$ has distance $3/4$ because the codewords agree at the fourth character ($D$). The rate is $(\log 2)/4$.
\end{example}

\begin{example}
    The code $\{1234,5634,1637\}$ over alphabet $\{1,2,3,4,5,6,7\}$ has distance $1/2$ because every pair of codewords agree on 2 out of 4 indices. The rate is $(\log 3)/4$.
\end{example}

\paragraph{Codes in the context of our paper.}
In the context of our paper, codewords will, roughly, correspond to feasible sets. 
Recall that we measure the performance of an algorithm compared to some benchmark: hindsight-optimal ``prophet'' for prophet inequality and optimal adaptive policy for stochastic probing.
\begin{itemize}
    \item {\em High rate} means more codewords/feasible sets, which means that the algorithm has more uncertainty about which optimal set will be chosen by the benchmark.
    \item {\em High distance} means that if the algorithm ``goes for'' a sub-optimal set, it will have small intersection with the benchmark's optimal set.  
    \item The {\em alphabet size} roughly corresponds to $n$, the number of elements; since we parameterize our approximation by $n$ (specifically, we want $\Tilde{\Omega}(\log^2 n)$-hardness), a smaller alphabet is desirable.  
\end{itemize}
While we want all these properties simultaneously,  there are well known theorems in coding theory that establish fundamental limitations on the possible achievable tradeoffs between alphabet size, distance, and rate.

%\Aviad{TBD if time - this could unify some of the probabilistic arguments but maybe not quite worth it :\\ We now prove a simple bound on the tradeoff achievable by a random ECC. \begin{claim}\label{claim:ECC} \end{claim}}

\paragraph{Doubly-optional advanced concepts from coding theory.}
We conclude this brief section with two useful variants of ECCs that are conceptually related to our constructions. 
\begin{itemize}
    \item {\bf List-decoding codes:} Sometimes we can get around the alphabet-vs-distance-vs-rate tradeoff limits by relaxing the distance requirement to only demand that a codeword is far from {\em most} other codewords. In the context of communication, if a message-codeword has some errors, it can at least be decoded to some short list of nearby codewords (where the right message may be, e.g., inferred from context). In the context of our construction, the higher $1-1/L^2$ distance guarantee only holds for codewords that are not siblings in the rooted tree, so with respect to this higher distance, the sibling-codewords correspond to the decoding list.
    
    \item  {\bf Tree codes:} The codewords of tree codes~\cite{Schulman93,Schulman96} are defined on the root-leaf paths of a rooted tree. Their distance guarantee between two codewords is proportional to the length from the latest common ancestor of the corresponding paths to the leaves. In the original context of communication, they are used to robustify interactive communication against errors. 
\end{itemize}

\section{$\Tilde{\Omega}(\log^2 n)$ Hardness for Prophet Inequality} \label{sec:prophet}

We consider the prophet inequality problem with a downward-closed feasibility constraint and a linear objective \cite{Rubinstein16}. 

\begin{defn}[Prophet Inequality]
We are given a downward-closed feasibility constraint  $\mathcal{F}$ on $n$ elements. Each element has a non-negative random value drawn independently (but not necessarily identically) from a known distribution. The realization of the random values are sequentially revealed from a known fixed adversarial order, and the goal is to design an online algorithm that immediately accepts/rejects the element while always being feasible in $\mathcal{F}$ and trying to maximize the sum of accepted element values. The benchmark is the expected hindsight optimum, i.e., the expected maximum total realized value of a feasible set. 
\end{defn}

The following is our main hardness result for prophet inequality.

\begin{theorem} \label{thm:PIHardness}
    For the prophet inequality problem with a downward-closed feasibility constraint and an additive objective, there exists an instance where every online algorithm is  $\Omega(L^2)$-competitive for $L= \Theta(\log n/\log\!\log n)$.
\end{theorem}

%\snote{Fix this para later based on intro} At a high-level, the reason $O(\log n)$-competitive algorithms are impossible for prophet inequality, unlike the secretary model, is because the elements with non-zero value might appear in increasing order of their values, which makes it impossible for the algorithm to union bound over all possible feasible subsets....

\subsection{Construction}

Let $L = \Theta(\log(n)/\log\log(n))$  such that $L^{2L} =o(n^{0.1})$.  We will provide a construction for which the expected hindsight optimum is $\Omega(L)$ but every online algorithm has expected total value $O(1/L)$, which will prove \Cref{thm:PIHardness}.

Define parameters $p := n^{0.3}$  and  $r_{\ell}:=L^{2\ell}$ for $\ell \in [L]$ (we will later need $r_L^3 = o(p)$).  Also, we set  $d_{\ell} := L^{2\ell-1} = r_{\ell} /L$ and $d_{(<\ell)} := L^{2\ell-2} = d_{\ell} /L = r_{\ell} /L^2$.

\paragraph{Basic Setup.}
We consider a rooted forest with $L+1$ \emph{layers}. For $\ell \in [L]$, each layer $(\ell-1)$ node will have $r_{L}^{d_{\ell}}= (L^{2L})^{d_{\ell}}$ children (the root is at layer $0$). 
Each node at layer $\ell \in [L]$  will correspond to a subset of elements (the root has no subset): the subsets for nodes at the same layer can intersect, but they are disjoint across layers. There is a correspondence between maximal feasible sets in $\mathcal{F}$ and root-leaf paths in this tree: on each path the union of subsets corresponding to all the vertices is feasible.  %correspond  is the union of subsets on all root-leaf paths in this forest. \Aviad{previous wording is confusing imho}

\paragraph{Layer-$\ell$ elements and arrival order.}
We use  $U_{\ell}$ to denote the set of all Layer-$\ell$ elements; we set $|U_{\ell}| = p^2\cdot r_{\ell}$. (Note that the total number of elements $\sum_\ell |U_\ell| \leq p^2 \cdot r_L \cdot L  < n$.)
Elements in 
$U_{\ell}$  each take  either value $0$ ({\em inactive}) or $1/r_{\ell}$ ({\em active}).
Every element, regardless of layer, is active with probability $2/L^2$.
Every layer-$\ell$ subset has size $r_{\ell}$.
Notice that the maximum total value for a layer-$\ell$ subset is $1$ (in case all elements are active).

  The elements arrive in reverse order of layers, i.e.~from the layer of leaves to the layer of root; the elements within a layer arrive in an arbitrary order.

\paragraph{Layer-$\ell$ subsets - double error correcting codes.}
Recall that $|U_{\ell}| = p^2\cdot r_{\ell}$, which we think of as $r_{\ell}$ {\em blocks} of $p^2$ elements. Each feasible vertex-subset will have exactly one element from each bucket. We can naturally identify between node subsets and pairs of  length-$r_{\ell}$ vectors over $[p]$: for $i\in [r_\ell]$: the subset has one of $p^2$ elements in the $i$-th bucket, encoding the respective $i$-th entries of both vectors. 
See also Figure~\ref{fig:Double-ECC} for an example.

\begin{figure*}[t!]
\centering
\begin{tikzpicture}[every node/.style={draw, rounded corners, align=center, font=\small}]
  % root and first-level nodes
  \node (root) at (0,0) {$\emptyset$};
  \node (left) at (-3,-1.5) {$\{a,b\}$};
  \node (right) at (3,-1.5) {$\{b,c\}$};
  \draw (root) -- (left);
  \draw (root) -- (right);

  % left subtree leaves (staggered diagonally)
  \node (l1) at (-5.2,-2.8) {$\{A1,\,B2,\,C3,\,D4\}$};
  \node (l2) at (-3.0,-3.5) {$\{A3,\,B4,\,C5,\,D6\}$};
  \node (l3) at (-0.8,-4.2) {$\{A1,\,B6,\,C3,\,D7\}$};
  \draw (left) -- (l1);
  \draw (left) -- (l2);
  \draw (left) -- (l3);

  % right subtree leaves (staggered diagonally)
  \node (r1) at (0.8,-2.8) {$\{G1,\,F2,\,E3,\,D4\}$};
  \node (r2) at (3.0,-3.5) {$\{G3,\,F4,\,E5,\,D6\}$};
  \node (r3) at (5.2,-4.2) {$\{G1,\,F6,\,E3,\,D7\}$};
  \draw (right) -- (r1);
  \draw (right) -- (r2);
  \draw (right) -- (r3);
\end{tikzpicture}

\raggedright
\small
At the leaves layer, the lower rate, higher distance ($3/4$) code is $\{ABCD,GFED\}$;
the higher rate, lower distance ($1/2$) code is $\{1234,5634,1637\}$. 
\caption{Example of double error-correcting code construction}
\label{fig:Double-ECC}
\end{figure*}

%\snote{A slightly simpler construction might be to use $|U_{\ell}| = p^2$ elements, where we don't need to partition into $r_\ell$ blocks and instead just select uniformly at random. Decide this later.}
Formally, for every layer $\ell \in [L]$, 
we define two {\em families} of  vectors:  all these vectors are in $[p]^{r_{\ell}}$ (i.e., have coordinates in $[p]$ and have length $r_{\ell}$), where
\begin{itemize}
    \item The first family has $r_{L}^{d_{(<\ell)}}= (L^{2L})^{d_{(<\ell)}}$ vectors, with each coordinate drawn uniformly and independently at random in $[p]$. 
    \item The second family has $r_{L}^{d_{\ell}}= (L^{2L})^{d_{\ell}}$ vectors, with each coordinate drawn uniformly and independently at random in $[p]$.  %{\color{red}This $r_{\ell}^{d_{\ell}}$ needs to be changed to $(L^{2L})^{d_{\ell}} = (L^2)^{r_\ell}$. Similarly change $r_{\ell}^{d_{(<\ell)}}$ above}
\end{itemize}
We define a feasible subset of $U_{\ell}$ for each pair of vectors from the respective families, i.e., the $i$-th element in the subset is the element corresponding to the pair of respective $i$-th entries from each vector. 
In \Cref{claim:ECC_PI}, we will show that this choice of parameters implies that any pair of layer-$\ell$ subsets share  $\leq d_{\ell}$ elements, and  if they have different parent nodes then  share  $\leq d_{(<\ell)}$ elements.

% \paragraph{Connecting the layers of the forest}
% Notice that choosing a Layer-$\ell$, subset corresponds to choosing two polynomials over  $\mathbb{F}_{p}$, of respective degrees $d_{(<\ell)},d_{\ell}$. Therefore the total number of Layer-$\ell$ subsets is $p^{d_{(<\ell)} + d_{\ell}} = p^{\Theta(L^{2\ell-1})} $. This is smaller than $p^{L^{2\ell}}$, the number of ways we can choose just the first-coordinate polynomial for layer $\ell+1$. 

% We construct an arbitrary mapping from layer-$\ell$ nodes to first-coordinate polynomial for layer-$(\ell+1)$, and create an edge between each layer-$\ell$ node and each of the layer-$(\ell+1)$ nodes that share the matching first-coordinate polynomial. We discard the remaining unmatched layer-$(\ell+1)$ nodes. 

\paragraph{Connecting the layers of the forest.}
%Notice that, after fixing the families, choosing a layer-$\ell$, subset corresponds to choosing the two vectors, so a total of $p^{d_{\ell}+d_{(<\ell)}}$ choices; this is less than $p^{d_{<\ell+1}}$, \snote{check this} or the number of ways to choose just the first vector for layer $\ell+1$.

Notice that the number of nodes  at layer $\ell$ equals $(r_L)^{d_{<\ell+1}}$, which is at most number of first family vectors 
$(r_L)^{d_{(<\ell+1)}} $ at layer $\ell$ since
\begin{align}
\label{eq:ProphLBSimpleIneq}
d_{<\ell+1} := \sum_{\ell'<\ell+1} d_\ell &= \sum_{\ell'<\ell+1} L^{2\ell'-1} =  \frac{1}{L}\cdot \frac{L^{2\ell+2}-1}{L^2-1} \notag \\
&=\Theta(L^{2\ell-1})\leq  L^{2\ell}= d_{(<\ell+1)}.
\end{align}

Thus, we construct an arbitrary injective mapping from layer-$(\ell-1)$ nodes to first-family vectors for layer-$\ell$, and create an edge between each layer-$(\ell-1)$ node and each of the layer-$\ell$ nodes that share the matching first-family vector. We discard the remaining unmatched layer-$\ell$ nodes.

\subsection{Proof of \Cref{thm:PIHardness}}

We prove that the expected hindsight optimum is $\Omega(L)$ but every online algorithm has expected total value $O(1/L)$, which implies $\Omega(L^2)$ hardness on competitive ratio.

\begin{lemma}[Completeness: Lower bound on \Opt]\label{lemma:PIAlgLB}
The expected value of the hindsight optimum is $\Omega(L)$.
    %With high probability, there is a feasible set where every element is active (hence the total value is $L$).
\end{lemma}

\begin{proof}
We construct an optimal solution by traversing a root-leaf path. 
At each step, we have $(r_L)^{d_{\ell}} $ choices for the next node, each corresponding to a vector from the second family (which consists of  independently drawn length-$r_{\ell}$ vectors with coordinates  in $[p]$).

For any fixed choice of first-family vector and every entry in the corresponding second-family vectors, we have that with very high probability,  at least $1/L^2$-fraction of all corresponding $U_{\ell}$ elements are active  (since activation probabilities were $2/L^2$). Thus the probability that a randomly drawn second-family vector corresponds to $\Omega(r_\ell)$ active elements is at  least $(L^{-2})^{r_{\ell}/2} = (L^{-r_{\ell}}) = L^{-L \cdot d_{\ell}}$.
Since each node has $(L^{2L})^{d_{\ell}} \gg 1/\Pr[\text{$\Omega(r_\ell)$ active elements}]$ children, we have that, for every layer, with high probability one of the children of the current node has $\Omega(r_\ell)$ active elements. %{\color{red}To make this last step work, we need to choose $ L^{-2L \cdot d_{\ell}}$ vectors (which is more than what's currently written)}
\end{proof}

\begin{lemma}[Soundness: Upper bound on ALG]\label{lemma:PIAlgUB}
    Any algorithm has total expected value at most $O(1/L)$.
\end{lemma}

The proof of this lemma relies on the following crucial claim about the size of intersection between any two subsets.
\begin{claim}\label{claim:ECC_PI}
    With high probability, every pair of layer-$\ell$ subsets share at most $d_{\ell}$ elements. Furthermore, if they don't agree on the first-family vector (i.e., don't have the same parent node), they share at most $d_{(<\ell)}$ elements.
\end{claim}

Before proving \Cref{claim:ECC_PI}, we complete the proof of the lemma.

\begin{proof}[Proof of~\cref{lemma:PIAlgUB}]
    Consider the first layer (starting from leaf-layer)  $\ell$ where the algorithm selects more than $d_{(<\ell)}$ elements. 
We now bound the total expected contribution from this layer, as well as the ones above and below in the forest. 

    \emph{Layers below}:
    Since $\ell$ is first, from each layer $\ell'>\ell$ the algorithm selected at most $d_{<\ell'}$ elements, by definition, for a total value of at most $1/L^2$ from each layer.

    \emph{Layer $\ell$}: 
    In this layer, the algorithm may be able to continue to select elements up until $d_{\ell}$ for a total value of $1/L$. From the moment it crosses $d_{\ell}$ elements, it has committed to a single subset in this layer by \Cref{claim:ECC_PI}, of which in expectation only $2/L^2$-fraction of elements will be active.

    \emph{Layers above}:
    Once the algorithm selected more than $d_{(<\ell)}$ elements from layer $\ell$, by \Cref{claim:ECC_PI} there is only one path from root to the single corresponding layer-$(\ell-1)$ node. On each node on this path, at most $2/L^2$-fraction of the elements will be active, in expectation. Thus the total expected value from layers $\ell' < \ell$ is also at most $2/L$. 

    This completes the proof of \Cref{lemma:PIAlgUB}, except the missing \Cref{claim:ECC_PI}.
    \end{proof}

\paragraph{Distance guarantee of the random (double) error correcting code}

\begin{proof}[Proof of \Cref{claim:ECC_PI}]
We prove that, w.h.p., every pair of vectors in the first family agree on at most $d_{(<\ell)}$; the proof for the second family is identical (replacing $d_{(<\ell)}$ with $d_{\ell}$) .

Consider a vector $A$ in the first family, and draw the entries of the second vector $B$ one-by-one. Each one has probability $1/p$ of matching with the corresponding entry of $A$.  Thus the probability that a particular subset of $d_{(<\ell)}$ entries from $B$ agrees with $A$ is  $(1/p)^{d_{(<\ell)}}$. Taking a union bound over all $d_{(<\ell)}$-subsets of $B$'s entries, and all pairs $A,B$ in the first family, we have that the probability that any two vectors agree on more than $d_{(<\ell)}$ entries is at most  
$$
(1/p)^{d_{(<\ell)}} \times {{r_{\ell}}\choose{d_{(<\ell)}}} \times {{r_{L}^{d_{(<\ell)}}}\choose{2}} < (r_{L}^{3}/p)^{d_{(<\ell)}},
$$   
which is negligible by our setting of parameters as $r_{L}^3\ll p$.
\end{proof}

Putting \Cref{lemma:PIAlgLB} and \Cref{lemma:PIAlgUB} completes the proof of \Cref{thm:PIHardness}.

\section{$\Tilde{\Omega}(\log^2 n)$  Adaptivity Gap for Stochastic Probing} \label{sec:probing}

We consider the stochastic probing problem, as originally defined by Gupta and Nagarajan \cite{GN-ICALP13}, and then subsequently studied in several works such as \cite{LiLZ25,KMS-STOC24,PRS-APPROX23,BSZ-RANDOM19,GNS-SODA17}.

\begin{defn}[Stochastic Probing]
    We are given two downward-closed feasibility constraints on elements $[n]$: \emph{outer feasibility} constraints  $\mF_{\text{out}}$ and \emph{inner feasibility} constraints $\mF_{\text{in}}$. We are also given $n$ independent distributions, where the random value $X_i$ of $i$-th element is drawn from the $i$-th distribution. The goal is to design an algorithm that  \emph{probes} (i.e., find random realization) a subset $S$ of elements that are feasible in the outer feasibility constraint, i.e. $S \in \mF_{\text{out}}$, while maximizing its total value, $\max_{T \in \mF_{\text{in}}} \sum_{i \in T\cap S} X_i$, in expectation.
\end{defn}

This problem admits two natural classes of algorithms: an \emph{adaptive} algorithm sequentially probes the elements, where the decision of which element to probe next depends on the realizations of the previously probed elements. On the other hand, a \emph{non-adaptive} algorithm must upfront decide a subset of elements $S \in \mF_{\text{out}}$ to probe, and its value is $\E[\max_{T\in \mF_{\text{in}}} \sum_{i\in S \cap T} X_i]$. 

For any instance of stochastic probing, clearly  the value of the optimal adaptive algorithm is at least the  value of the optimal non-adaptive algorithm. We say that a  stochastic probing instance has $\alpha$ \emph{adaptivity gap}  if the value of the optimal adaptive algorithm is at most $\alpha$ times the value of the optimal non-adaptive algorithm (in expectation).

\begin{theorem} \label{thm:adapGapLB}
There exists an instance of the stochastic probing problem with general down-closed outer and inner constraints where the adaptivity gap is $\Omega(L^2)$ for $L= \Theta(\log n/\log\!\log n)$.
\end{theorem}

\subsection{Construction}
Let $L= \Theta(\log n/\log\!\log n)$, where $L^{2L} =o(n^{0.1})$. We  present a randomized construction where the expected value of the optimal adaptive policy is $\Omega(L)$ but every non-adaptive policy has expected value $O(1/L)$, showing an $\Omega(L^2)$ adaptivity gap.

Define parameters  $p := n^{0.4}$ and $r_{\ell}:=L^{2\ell}$ for $\ell \in [L]$  (we will later need $r_L^4 = o(p)$). Also, we set parameters 
\[ d_{\ell} := L^{2\ell-1} = r_{\ell} /L \qquad \text{and} \qquad d_{(<\ell)} := L^{2\ell-2} = d_{\ell} /L = r_{\ell} /L^2.\]

%Define parameters $p := n^{0.3}$  and  $r_{\ell}:=L^{2\ell}$ for $\ell \in [L]$ (we will later need $r_L^3 = o(p)$).  Also, we set  $d_{\ell} := L^{2\ell-1} = r_{\ell} /L$ and $d_{(<\ell)} := L^{2\ell-2} = d_{\ell} /L = r_{\ell} /L^2$.  

%The construction will involve a rooted tree on ``blocks'', where each block itself is a rooted-tree  with edges corresponding to active/inactive elements. Thus, one could view the entire construction as a single \emph{combined rooted tree} with edges corresponding to elements.

\paragraph{Basic setup.}
We consider a rooted tree (the {\em meta-tree}) with $L$ \emph{layers} (starting from root at layer $1$). Each layer will correspond to a different range of values for ``active'' elements, where layer $1$ has the highest values. Each node at layer $\ell$ will correspond to a \emph{block}, which is an $L^2$-ary rooted tree  of depth $r_\ell$ (thus, a block has $(L^2)^{r_\ell}$ leaves). We will use \emph{layer} $\ell$ for level of  a block in the meta-tree and use \emph{height} $h$ for level of an element inside a block (where $h=1$ represent edges out of the root of the block). The combined tree consisting of replacing each node (+ edges to children) of the meta tree with the corresponding block is called the {\em concatenated tree}. 

\paragraph{Comparison with prophet inequality construction.}
The meta-tree for stochastic probing is analogous to the rooted tree for prophet inequality. But the inner structure of the blocks is different:  for stochastic probing we have rooted trees, whereas for prophet inequality we use independently chosen codewords that would be roughly analogous to using disjoint paths instead of a tree for stochastic probing. 

In the double error correcting code construction, this requires us to replace one of the error correcting codes (the lower distance code that differentiates between siblings) with a code defined on the rooted tree, resembling tree codes~\cite{Schulman93,Schulman96}. 

\paragraph{Feasible Sets.} The outer feasible sets will correspond to the \emph{caterpillar} subgraphs along  root-leaf paths of the concatenated tree, i.e., a subset of  edges is feasible to probe iff there exists a root-leaf path such that the probed edges have at least one end-point on the path. The inner feasibility constraints will correspond to root-leaf paths in the concatenated tree. The value of elements in layer $\ell$ will be either $0$ or $1/r_\ell$, depending on whether the element is inactive or active.

\paragraph{Construction of blocks.}
%We consider a rooted tree (the {\em meta-tree} with $L$ \emph{layers} (starting from root at layer $1$. Each layer will correspond to a different range of values for active elements, where layer $1$ has the highest values. Each node at layer $\ell$ will correspond to a \emph{block}, which is an $L^2$-ary rooted tree  of depth $r_\ell$ (thus, a block has $(L^2)^{r_\ell}$ leaves). We will use \emph{layer} $\ell$ for level of  a block in the meta-tree and use \emph{height} $h$ for level of an element inside a block (where $h=1$ represent edges out of the root of the block). The combined tree consisting of replacing each node (+ edges to children) of the meta tree with the corresponding block is called the {\em concatenated tree}. 

Every edge in a block corresponds to some element. 
Blocks at different layers don't share any elements. Edges at different heights (irrespective of their layer and block) also don't share any elements. However, edges at the same layer and height may share elements. Let $U_{\ell,h}$, where $|U_{\ell,h}| = p^2$, denote the universe of elements at height $h$ edges of layer $\ell$. 
(Note that the total number of elements $\sum_{\ell,h} |U_{\ell,h}| \leq p^2 \cdot r_L \cdot L  < n$.)
%Blocks at different layers don't share any elements but blocks at the same layer $\ell$ have elements  from the same universe  $U_{\ell}$,  where $|U_{\ell}| = p^2$.
Each element in $U_{\ell,h}$ has  either value $0$ ({\em inactive}) or $1/r_{\ell}$ ({\em active}). 
Every element, regardless of its layer and height, is active with probability $2/L^2$.
Note that this implies that an all-active root-leaf path in any block has total value $r_\ell\cdot 1/r_\ell  = 1$.

%a rooted tree of depth $r_\ell$ and  with $(L^2)^{r_\ell}$ leaves, as defined above. 
 
%\paragraph{Construction of a block:} We will have $L$ layers of blocks where a \emph{block} at layer $\ell$ is a $L^2$-ary rooted tree with $r_\ell$ \emph{levels}. (We will use ``layers'' for block-layers and ``levels'' for element-levels inside a block.) 

\paragraph{Assignment of elements to blocks (using double ECC).}
%\paragraph{Layer-$\ell$ subsets - double error correcting codes}
Elements at layer $\ell$ and height $h$ belong to $U_{\ell,h}$,  where $|U_{\ell,h}| =p^2$, and each element should be thought of as corresponding to a tuple in $[p]\times [p]$. Thus, to specify an element for an edge at layer $\ell$ and height $h$, it suffices to assign a pair of numbers from $[p]$ to it. %To assign these tuples, we define the following vectors.

% %of a block at layer $\ell$ belong to $U_{\ell}$, where $|U_\ell| =p^2$, and each element should be thought of as corresponding to tuple in $[p]\times [p]$. Thus, to specify an element of a block, it suffices to assign a tuple of numbers in $[p]$ to it. To assign these tuples, we define the following vectors.

% For each layer $\ell \in [L]$, we 
% define a family of vectors of length $r_{\ell}$ over $[p]$, and we define a single vector of length $(L^2)^{r_\ell}$ (which is the same as the number of leaves of a layer-$\ell$ block) over $[p]$:
\begin{itemize}
    \item Take a family of $(r_L)^{d_{(<\ell)}} $ vectors of length $r_{\ell}$ drawn uniformly and independently at random. %\snote{Or shall we go for just $(r_L)^{d_{<\ell}} = (L^2)^{r_{<\ell}} $ vectors to avoid confusion?}
    (These vectors will be assigned to the blocks at layer $\ell$, and the $h$-th coordinate, where $h\in [r_\ell]$, of the vector will correspond to height $h$ edges  of the block.) 
    
    \item For each edge $e$ in a layer-$\ell$ block, draw an independently random $x_e \in [p]$. 
        % Next, draw a single vector of length $(L^2)^{r_\ell} = r_L^{d_\ell}$ uniformly and independently at random. (For any height $h$ of a block, the $i$-th coordinate will correspond to the $i$-th edge at that height.)
\end{itemize}
%We define a block for each vector from the family: i.e.,  the $i$-th element at height $h$ of the block at layer $\ell$ will correspond to the $h$-th coordinate of the first vector and the $i$-th coordinate of the second vector.
For each edge at height $h$, we assign the element of $U_{\ell,h}$ whose first coordinate corresponds to the $h$-th coordinate of the corresponding vector in the first family, and whose second coordinate is $x_e$.

\paragraph{Connecting the layers of the forest.}
We observe that the number of blocks needed at layer $\ell$ is  at most the number of vectors in the family at level $\ell$, since 
\[(L^2)^{r_1 + r_2+ \ldots +r_{\ell-1}} = (L^{2L})^{d_1 + d_2+ \ldots +d_{\ell-1}} \leq (L^{2L})^{d_{(<\ell)}} = (r_L)^{d_{(<\ell)}} ,\]
where the inequality uses
\[
    \sum_{\ell'<\ell} d_{\ell'} = \sum_{\ell'<\ell} L^{2\ell'-1} =  \frac{1}{L}\cdot \frac{L^{2\ell}-1}{L^2-1} =\Theta(L^{2(\ell-1)-1})\leq  L^{2\ell-2}= d_{(<\ell)}.
\]
Hence, we can construct an arbitrary one-to-one mapping from blocks at layer-$\ell$  to the vectors in the family (dropping any extra vectors). %\snote{This step simplifies if we take  exactly $(r_L)^{d_{<\ell}} = (L^2)^{r_{<\ell}}$ vectors.}

%We will later show that any pair of root-leaf paths inside a block at layer-$\ell$ share at most $d_{\ell}$ elements beyond their common prefix. Moreover,  any pair of root-leaf paths in two distinct blocks at the same layer-$\ell$ (i.e., they disagree on the vector in the family) share at most $d_{(<\ell)}$ elements. 

% All elements in a block (corresponding to edges) share the same first number (drawn uniformly in $[p]$). The second number of each element in the block is drawn independently and uniformly at random from $[p]$. Since this is an $L^2$-ary tree of height $r_\ell$, the number of leaves of this block are $(L^2)^{r_\ell}$.
%Each edge in the block is assigned a tuple in $[p]\times [p]$ independently and uniformly at random. 

%Number of potential root-leaf path labelings is $p^{2 r_\ell}$. 

%We construct an arbitrary one-to-one mapping from leaves of layer-$\ell$ blocks (in total there are $(L^2)^{r_1 + r_2+\ldots+r_\ell}$ leaves) to first-family vectors for layer-$(\ell+1)$ blocks, and create an edge between each layer-$\ell$ node and each of the layer-$(\ell+1)$ nodes that share the matching first-family vector. We discard the remaining unmatched layer-$(\ell+1)$ nodes. 

%Notice that, after fixing the families, choosing a layer-$\ell$ subset corresponds to choosing the two vectors, so a total of $p^{d_{\ell}+d_{(<\ell)}}$ choices; this is less than $p^{d_{<\ell+1}}$, or the number of ways to choose just the first vector for layer $\ell+1$.

\subsection{Proof of \Cref{thm:adapGapLB}}

We will show that the expected value of the optimal adaptive policy is $\Omega(L)$ but every non-adaptive policy has expected value $O(1/L)$; hence, we will have an $\Omega(L^2)$ adaptivity gap.

\begin{lemma}[Completeness: Lower bound on \Opt] \label{lem:adapGapOPT}
The expected value of the optimum adaptive solution is $\Omega(L)$, where the expectation is both over the random construction of the stochastic probing instance and the randomness of active/inactive elements.
\end{lemma}

\begin{proof}
We construct an adaptive solution by traversing a root-leaf path: Starting at the root node,  at each node the algorithm probes all the incident edges and then  takes an arbitrary active edge  or an arbitrary edge if all of them are inactive. We will argue that the expected value from each layer $\ell \in [L]$ is $\Omega(1)$, so the total expected value is $\Omega(L)$.

Consider any block at layer $\ell$ and consider the edges at height $h$. 
All possible elements at this height $U_{\ell,h}$  have been assigned to edges uniformly at random. 
%All elements in this block belong to $U_\ell$ and have been assigned to edges uniformly at random. 
Since each element is active w.p. $2/L^2$, w.h.p. $\Theta(1/L^2)$ fraction of $U_{\ell,h}$ elements are active. Thus, the chance that the current node sees an active outgoing edge is $1-(1-\Theta(1/L^2))^{L^2} = \Omega(1)$. Hence, the expected number of active edges on the root-leaf path of a block at layer $\ell$ is $\Omega(r_\ell)$, which means the total value from this block is $\Omega(1)$ as each active edge at layer $\ell$ has value $1/r_\ell$. 
\end{proof}

\begin{lemma}[Soundness: Upper bound on ALG]\label{lem:adapGapALG}
    The generated random instance will satisfy w.h.p.\,that  every non-adaptive algorithm has total expected value at most $O(1/L)$.
\end{lemma}

\begin{proof} The proof relies on the following claim regarding the number of intersecting elements between paths and caterpillars inside blocks. (Recall that different layers have no intersection.)
    
\begin{claim}\label{claim:ECC_AdapGap}
For any layer $\ell$,  w.h.p., we have the following two properties:
    \begin{itemize}
        \item Different blocks: Any root-leaf path and any caterpillar inside two different blocks have at most $d_{(<\ell)}$ heights where they share an element. (Observe that we are not counting how many intersections happen at the same height.)

        \item Same block: Any root-leaf path $P$ and caterpillar $C$ inside the same block share at most $d_{\ell}$ elements beyond their common prefix.
    \end{itemize}
    \end{claim}
    
% \begin{proof}
%     The heart of the proof lies in the following claim, which is the only part of  the proof where we need the randomness of the generated instance.
%     \begin{claim}\label{claim:ECC_AdapGap}
%     With high probability, every pair of root-leaf paths inside a block at layer-$\ell$ share at most $d_{\ell}$ elements beyond their common prefix.
%     Furthermore, if they are not in the same block  (i.e., they disagree on the vector in the family), they  share at most $d_{(<\ell)}$ elements.     
%     \end{claim}

%\snote{Better to have different $U_\ell$ base elems for each level. The revised claim is to bound the intersection of any caterpillar and any root-leaf path (only after the common path P since anyways the universe $U_\ell$ is different across levels). }

    Before proving the claim, we use it to complete the proof of the lemma. We will show that for any random instance given by \Cref{claim:ECC_AdapGap},  every non-adaptive solution has expected value $O(1/L)$.
    
    Consider any non-adaptive solution, which is a caterpillar $C$ along a root-leaf path (spine) in the concatenated tree. First, the expected value from the spine of this path is $O(1/L)$, since each element is active only with probability $1/L^2$. Next, we argue that even if the non-adaptive solution considers a different root-leaf path  $P$ for its value in the inner constraint (recall that each root-leaf path is feasible in the inner constraint and our objective is to bound $\E\big[\max_P [\cdot] \big]$, not $\max_P \big[\E [\cdot ]\big]$), it can receive at most $O(1/L)$ additional value.

    Towards proving this, consider the first layer $\ell$ (starting from root at layer $1$)  where path $P$ deviates from the caterpillar $C$.
    We consider the three cases:
    \begin{enumerate}
        \item From layers  $<\ell$  (i.e., above $\ell$), the total value is $O(1/L)$ since we already accounted for it in the spine.

        \item For each layer  $\ell'>\ell$ (i.e., below $\ell$), since the blocks of the spine and $P$ are different, there are at most $d_{(<\ell')}$ levels where they may share an element  by \Cref{claim:ECC_AdapGap}, which means the total value contribution from this block is at most $d_{(<\ell')}/r_{\ell'} =  1/L^2$. Summing over all layers $\ell'>\ell$, the sum is $O(1/L)$.

        \item From layer $\ell$, the caterpillar $C$ and path $P$ overlap until some node in this block. The value above this node in the block is still $O(1/L)$ since we still follow the spine. For the total value below, we know that path $P$ and  caterpillar $C$ share at most $d_{\ell}$ elements inside this block by \Cref{claim:ECC_AdapGap}, which means their total value contribution is at most $d_{\ell}/r_\ell = O(1/L)$.    \qedhere
    \end{enumerate}
\end{proof}

Finally, we prove the missing \Cref{claim:ECC_AdapGap}, which can be thought of as some kind of distance guarantee of the random (double) error correcting code.

\begin{proof}[Proof of \Cref{claim:ECC_AdapGap}] We prove the claim separately for different blocks and the same block.

    \textbf{Different blocks.} Consider any  root-leaf path and a caterpillar in two different blocks. Note that for them to agree on $d_{(<\ell)}$ elements, the  two vectors from the family that correspond to these two blocks have to agree on at least $d_{(<\ell)}$ coordinates since elements from different levels appear from disjoint universe.  
    We prove below that, w.h.p., every pair of vectors in the  family  agree on at most $d_{(<\ell)}$ coordinates.

Consider any vector $A$ in the family, and draw the entries of the second vector $B$ one-by-one.
Each one has probability $1/p$ of matching the corresponding entry of $A$.
Thus the probability that a particular subset of $d_{(<\ell)}+1$ entries from $B$ agrees with $A$ is
$(1/p)^{d_{(<\ell)}+1}$.
Taking a union bound over all $(d_{(<\ell)}+1)$-subsets of $B$'s entries, and all pairs $A,B$ in the family, we get
\[
(1/p)^{d_{(<\ell)}+1}
\times
\binom{r_\ell}{d_{(<\ell)}+1}
\times
\binom{r_L^{d_{(<\ell)}}}{2}
<
(r_L^3/p)^{d_{(<\ell)}+1},
\]
which is negligible since $r_L^3 \ll p$.

\textbf{Same block.} Similarly, now consider any root-leaf path $P$ and a caterpillar $C$ in the same block. 
For them to agree on $d_{\ell}$ elements beyond their common prefix,  the corresponding two vectors from the second vectors have to agree on at least $d_{\ell}$ coordinates. 

Consider the second vector $A$ (i.e., not the  vector in family) corresponding to any path $P$. Now draw  entries of the caterpillar one-by-one (note that a caterpillar of length $r_\ell$ has at most $(r_\ell \cdot L^2)$  elements). Each entry has probability at most $1/p$ of matching with the corresponding entry of $A$. 
Thus, the probability that a particular subset of $d_\ell+1$ entries from the caterpillar agrees with $A$ is at most
$(1/p)^{d_\ell+1}$.
Taking a union bound over all $(d_\ell+1)$-subsets of $C$'s entries, and all pairs of paths $P$ and caterpillars $C$, we get
\[
(1/p)^{d_\ell+1}
\times
\binom{L^2 r_\ell}{d_\ell+1}
\times
(r_L^{d_\ell})^2
<
(r_L^4/p)^{d_\ell+1},
\]
which is negligible since $r_L^4 \ll p$.
\end{proof}

Putting \Cref{lem:adapGapOPT} and \Cref{lem:adapGapALG} completes the proof of \Cref{thm:adapGapLB}.

\clearpage
%\begin{small}
\bibliographystyle{alpha}
\bibliography{bib}

\newcommand{\etalchar}[1]{$^{#1}$}
\begin{thebibliography}{CHMS10}

\bibitem[AKW19]{AzarKW19}
Pablo~Daniel Azar, Robert Kleinberg, and S.~Matthew Weinberg.
\newblock Prior independent mechanisms via prophet inequalities with limited
  information.
\newblock {\em Games Econ. Behav.}, 118:511--532, 2019.

\bibitem[ANS08]{ANS-WINE08}
Arash Asadpour, Hamid Nazerzadeh, and Amin Saberi.
\newblock Stochastic submodular maximization.
\newblock In {\em Proceedings of 4th International Workshop on Internet and
  Network Economics, {WINE}}, volume 5385, pages 477--489. Springer, 2008.

\bibitem[BIK07]{BabaioffIK07}
Moshe Babaioff, Nicole Immorlica, and Robert Kleinberg.
\newblock Matroids, secretary problems, and online mechanisms.
\newblock In {\em Proceedings of the Eighteenth Annual {ACM-SIAM} Symposium on
  Discrete Algorithms, {SODA}}, pages 434--443. {SIAM}, 2007.

\bibitem[BIKK18]{babaioff2018matroid}
Moshe Babaioff, Nicole Immorlica, David Kempe, and Robert Kleinberg.
\newblock Matroid secretary problems.
\newblock {\em Journal of the ACM (JACM)}, 65(6):1--26, 2018.

\bibitem[BSZ19]{BSZ-RANDOM19}
Domagoj Bradac, Sahil Singla, and Goran Zuzic.
\newblock (near) optimal adaptivity gaps for stochastic multi-value probing.
\newblock In {\em Proceedings of {APPROX/RANDOM}}, pages 49:1--49:21, 2019.

\bibitem[CHMS10]{CHMS-STOC10}
Shuchi Chawla, Jason~D Hartline, David~L Malec, and Balasubramanian Sivan.
\newblock Multi-parameter mechanism design and sequential posted pricing.
\newblock In {\em Proceedings of the forty-second ACM symposium on Theory of
  computing, {STOC}}, pages 311--320, 2010.

\bibitem[CIK{\etalchar{+}}09]{CIKMR-ICALP09}
Ning Chen, Nicole Immorlica, Anna~R Karlin, Mohammad Mahdian, and Atri Rudra.
\newblock Approximating matches made in heaven.
\newblock In {\em International colloquium on automata, languages, and
  programming}, pages 266--278. Springer, 2009.

\bibitem[Dob07]{dobzinski2007two}
Shahar Dobzinski.
\newblock Two randomized mechanisms for combinatorial auctions.
\newblock In {\em Proceedings of {APPROX/RANDOM}}, pages 89--103. Springer,
  2007.

\bibitem[Dyn63]{Dynkin}
E.~B. Dynkin.
\newblock {The optimum choice of the instant for stopping a Markov process}.
\newblock {\em Soviet Math. Dokl}, 4, 1963.

\bibitem[EIV23]{EIV-Book23}
F.~Echenique, N.~Immorlica, and V.~V. Vazirani.
\newblock {\em Online and Matching-Based Market Design}.
\newblock Cambridge University Press, 2023.

\bibitem[Fre75]{Freedman1975}
David~A. Freedman.
\newblock On tail probabilities for martingales.
\newblock {\em Ann. Probab.}, 3(1):100--118, 02 1975.

\bibitem[GN13]{GN-ICALP13}
Anupam Gupta and Viswanath Nagarajan.
\newblock A stochastic probing problem with applications.
\newblock In {\em Proceedings of {IPCO}}, volume 7801, pages 205--216, 2013.

\bibitem[GNS16]{GNS-SODA16}
Anupam Gupta, Viswanath Nagarajan, and Sahil Singla.
\newblock Algorithms and adaptivity gaps for stochastic probing.
\newblock In {\em Proceedings of Symposium on Discrete Algorithms, {SODA}},
  pages 1731--1747, 2016.

\bibitem[GNS17]{GNS-SODA17}
Anupam Gupta, Viswanath Nagarajan, and Sahil Singla.
\newblock {Adaptivity Gaps for Stochastic Probing: Submodular and XOS
  Functions}.
\newblock In {\em Proceedings of Symposium on Discrete Algorithms, {SODA}},
  pages 1688--1702, 2017.

\bibitem[GS20]{GS-Book20}
Anupam Gupta and Sahil Singla.
\newblock Random-order models.
\newblock In Tim Roughgarden, editor, {\em Beyond the Worst-Case Analysis of
  Algorithms}, pages 234--258. Cambridge University Press, 2020.

\bibitem[KMS24]{KMS-STOC24}
Thomas Kesselheim, Marco Molinaro, and Sahil Singla.
\newblock Supermodular approximation of norms and applications.
\newblock In {\em Proceedings of Symposium on Theory of Computing, {STOC}},
  pages 1841--1852, 2024.

\bibitem[KS77]{KrengelS77}
U.~Krengel and L.~Sucheston.
\newblock Semiamarts and finite values.
\newblock {\em Bulletin of the American Mathematical Society}, 83:745--747,
  1977.

\bibitem[KS78]{KrengelS78}
U.~Krengel and L.~Sucheston.
\newblock On semiamarts, amarts, and processes with finite value.
\newblock {\em Advances in Probability and Related Topics}, 4:197--266, 1978.

\bibitem[LLZ25]{LiLZ25}
Jian Li, Yinchen Liu, and Yiran Zhang.
\newblock Adaptivity gaps for stochastic probing with subadditive functions.
\newblock {\em CoRR}, abs/2504.15547, 2025.
\newblock To appear in FOCS 2025.

\bibitem[Luc17]{Lucier17}
Brendan Lucier.
\newblock An economic view of prophet inequalities.
\newblock {\em SIGecom Exch.}, 16(1):24--47, 2017.

\bibitem[PRS23]{PRS-APPROX23}
Kalen Patton, Matteo Russo, and Sahil Singla.
\newblock Submodular norms with applications to online facility location and
  stochastic probing.
\newblock In {\em Proceedings of {APPROX/RANDOM}}, volume 275, pages
  23:1--23:22, 2023.

\bibitem[RS17]{RubinsteinS17}
Aviad Rubinstein and Sahil Singla.
\newblock Combinatorial prophet inequalities.
\newblock In {\em Proceedings of Symposium on Discrete Algorithms, {SODA}},
  pages 1671--1687. {SIAM}, 2017.

\bibitem[Rub16]{Rubinstein16}
Aviad Rubinstein.
\newblock Beyond matroids: secretary problem and prophet inequality with
  general constraints.
\newblock In {\em Proceedings of the 48th Annual {ACM} {SIGACT} Symposium on
  Theory of Computing, {STOC}}, pages 324--332. {ACM}, 2016.

\bibitem[Sch93]{Schulman93}
Leonard~J. Schulman.
\newblock Deterministic coding for interactive communication.
\newblock In {\em Proceedings of Symposium on Theory of Computing, {STOC}},
  pages 747--756. {ACM}, 1993.

\bibitem[Sch96]{Schulman96}
Leonard~J. Schulman.
\newblock Coding for interactive communication.
\newblock {\em {IEEE} Trans. Inf. Theory}, 42(6):1745--1756, 1996.

\end{thebibliography}
%\end{small}

\end{document}